\documentclass[10pt,english]{article}
\usepackage[margin=1in]{geometry}
\usepackage[british]{babel}
\usepackage[backend=bibtex8,firstinits=true,maxnames=2,citestyle=numeric-comp]{biblatex}
\usepackage{hyperref}
\usepackage{amsmath,amssymb,amsthm,amsfonts}
\usepackage{float}
\usepackage{xfrac}
\usepackage{tikz}
\usepackage{textcomp}
\usepackage{csquotes}
\usepackage{mathtools}
\usepackage[normalem]{ulem}
\usepackage{thmtools}
\usepackage{thm-restate}
\usepackage{authblk}
\usepackage{booktabs}
\usepackage{adjustbox}

\newcommand{\mbf}[1]{\boldsymbol{#1}}
\newcommand{\mcl}[1]{\mathcal{#1}}
\newcommand{\ta}{\star}
\newcommand{\uar}{\xleftarrow{\$}}
\newcommand{\thsp}{\hspace*{1pt}}
\newcommand{\ts}[1]{{\textstyle #1}}

\newenvironment{myfont}{\fontfamily{phv}\selectfont}{\par}
\DeclareTextFontCommand{\textmyfont}{\myfont}
\newcommand{\Gen}{\textmyfont{KeyGen}}
\newcommand{\Enc}{\textmyfont{Enc}}
\newcommand{\Dec}{\textmyfont{Dec}}

\newcommand{\Add}{\textmyfont{Add}}
\newcommand{\Mult}{\textmyfont{Mult}}
\newcommand{\sk}{\textmyfont{sk}}

\newtheorem{definition}{Definition}

\begin{document}

\title{Practical Homomorphic Encryption Over the Integers}
\author[1]{James Dyer\thanks{james.dyer@postgrad.manchester.ac.uk}}
\author[2]{Martin Dyer\thanks{m.e.dyer@leeds.ac.uk}}
\author[2]{Jie Xu\thanks{j.xu@leeds.ac.uk}}
\affil[1]{School of Computer Science, University of Manchester}
\affil[2]{School of Computing, University of Leeds}

\maketitle
\begin{abstract}
We present novel homomorphic encryption schemes for integer arithmetic, intended for use in secure single-party computation in the cloud. These schemes are capable of securely computing only low degree polynomials homomorphically, but  this appears sufficient for most practical applications. In this setting, our schemes lead to practical key and ciphertext sizes. We present a sequence of generalisations of our basic schemes, with increasing levels of security, but decreasing practicality. We have evaluated the first four of these algorithms by computing a low-degree inner product. The timings of these computations are extremely favourable. Finally, we use our ideas to derive a fully homomorphic system, which appears impractical, but can homomorphically evaluate arbitrary Boolean circuits.
\end{abstract}


\section{Introduction}

With services like Amazon's Elastic MapReduce 
and Microsoft's HDInsight 
offering large-scale distributed cloud computing environments, computation in the cloud is becoming increasingly more available. Such services allow for computation on large volumes of data to be performed without the large investment in local computing resources. However, where the data that is processed is sensitive, such as financial or medical data, then uploading such data in its raw form to such a third-party service becomes problematic.

To take advantage of these cloud services, we require a means to process the data securely on such a platform. We designate such a computation, \emph{secure computation in the cloud} (SCC). SCC should not expose input or output data to any other party, including the cloud service provider. Furthermore, the details of the computation should not allow any other party to deduce its inputs and outputs. Cryptography seems the natural approach to this problem.

However, it should be noted that van Dijk and Juels \cite{vandijk2010impossibility} show that cryptography alone cannot realise secure \emph{multi-party} computation in the cloud, where the parties jointly compute a function over their inputs while keeping their own inputs private. Since our approach is via homomorphic encryption, we will restrict our attention to what we will call \emph{secure single-party computation in the cloud} (SSCC).

\emph{Homomorphic encryption} (HE) seems to offer a solution to the SSCC problem. First defined by Rivest et al. \cite{rivest1978data} in 1978,  HE allows a function to be computed on encrypted inputs without ever decrypting the inputs. Suppose we wish to compute the function $f$ on inputs $x_1,x_2,\ldots,x_n$, then, under HE, $\Dec(f'(x_1',x_2',\ldots,x_n'))=f(x_1,x_2,\ldots,x_n)$, where $x_1',\ldots,x_n'$ are the encryptions of $x_1,\ldots,x_n$, $f'$ is the equivalent of $f$ in the ciphertext space, and $\Dec$ is the decryption function. One can easily see that HE would satisfy some of the requirements for secure computation in the cloud. A \emph{somewhat HE} scheme (SWHE) is an HE scheme which is HE for only limited imputs and functions.

\emph{Fully HE} (FHE) is an HE scheme that is homomorphic for all $f$. This was first realised by Gentry in 2009 \cite{gentry2009fully}, and appears to be the ideal HE scheme. However, despite the clear advantages of FHE, and many significant advances \cite{brakerski2011efficient,brakerski2011ring,brakerski2012leveled,brakerski2012fully}, it remains largely impractical. The two implementations of recent FHE schemes, HELib \cite{halevi2015bootstrapping} and FHEW \cite{ducas2015bootstrapping}, both perform very poorly in comparison with operations on unencrypted data, in their running time and space requirements. It is reported that a HELib implementation of the AES-128 circuit processed inputs in just over four minutes \cite{halevi2015helib}. Similarly, FHEW processed a single homomorphic NAND operation followed by a re-encryption in 0.69s and using 2.2GB of RAM. The paper~\cite{naehrig2011can} attempted to assess the practicality of one of the underlying SWHE schemes\cite{brakerski2011ring}, but with no positive conclusion.

Therefore, we take the view in this paper that only SWHE is currently of practical interest. Our goal is to develop new SWHE schemes which are practically useful, and which we have implemented, though we conclude the paper by showing that our ideas can be used to develop a (fairly impractical) FHE scheme.

\subsection{Scenario}
As introduced above, our work concerns secure single-party computation in the cloud. In our scenario, a secure client wishes to compute a function on a large volume of data. This function could be searching or sorting the data, computing an arithmetic function of numeric data, or any other operation. For the most part, we consider here the case where the client wishes to perform arithmetic computations on numeric data. This data might be the numeric fields within a record, and the non-numeric fields would be treated differently.

The client delegates the computation to the cloud. However, while the data is in the cloud, it could be subject to snooping, including by the cloud provider. The client does not wish to expose the input data, or the output of the computation, to possible snooping in the cloud. A snooper here will be a party who may observe the data and the computation in the cloud, but cannot, or does not, change the data or insert spurious data. (In our setting data modification would amount to pointless vandalism.) The snooping could be casual, simply displaying an uninvited interest, or malicious, intending to use the data for the attacker's own purposes.

To obtain the required data privacy, the client's function will be computed homomorphically, on an encrypted version of the data. The client encrypts the source data using a secret key and uploads the encrypted data to the cloud, along with a homomorphic equivalent of the target computation. The cloud environment performs the homomorphic computation on the encrypted data. The result of the homomorphic computation is then returned to the client, who decrypts it using the secret key, and obtains the output of the original computation.

In this scenario, we observe that the source data is never exposed in the cloud, but encryptions of the source data are. A snooper may observe the computation of the equivalent homomorphic function in the cloud environment. As a result, they may be able to deduce what operations are performed on the data, even though they do not know the inputs. A snooper may also be able to inspect the (encrypted) working data generated by the cloud computation, and even perform side computations of their own on this data. However, snoopers have no access to the secret key, so cannot make encryptions of their own to deduce the secret key.

\subsection{Definitions and Notation}
$x \uar S$ denotes a value $x$ chosen uniformly at random from the discrete set $S$.

$\Gen : \mathcal{S} \rightarrow \mathcal{K}$ denotes the key generation function operating on the security parameter space $\mathcal{S}$ and whose range is the secret key space $\mathcal{K}$.

$\Enc : \mathcal{M} \times \mathcal{K} \rightarrow \mathcal{C}$ denotes the symmetric encryption function operating on the plaintext space $\mathcal{M}$ and the secret key space $\mathcal{K}$ and whose range is the ciphertext space $\mathcal{C}$.

$\Dec : \mathcal{C} \times \mathcal{K} \rightarrow \mathcal{M}$ denotes the symmetric decryption function operating on the ciphertext space $\mathcal{C}$ and the secret key space $\mathcal{K}$  and whose range is the plaintext space $\mathcal{M}$.

$\Add : \mathcal{C} \times \mathcal{C} \rightarrow \mathcal{C}$ denotes the homomorphic addition function whose domain is $\mathcal{C}^2$ and whose range is $\mathcal{C}$.

$\Mult : \mathcal{C} \times \mathcal{C} \rightarrow \mathcal{C}$ denotes the homomorphic mutliplication function whose domain is $\mathcal{C}^2$ and whose range is $\mathcal{C}$.

$m,m_1,m_2,\ldots$ denote plaintext values. Similarly, $c,c_1,c_2,\ldots$ denote ciphertext values.

If $k^*=\binom{k+1}{2}$, $\mbf{v}_{\ta}=[v_1\ v_2\ \ldots\ v_{k^*}]^T$ denotes a $k^*$-vector which augments the $k$-vector $\mbf{v}=[v_1\ v_2\ \ldots\ v_k]^T$ by appending elements $v_i = f_i(v_1,\dots,v_k)$ $(i \in [k+1,k^*])$, for a linear function $f_i$. (All vectors are column vectors throughout.)

$\mbf{e}_i$ denotes the $i$th unit vector $(i=1,2,\ldots)$, with size determined by the context.

$[x,y]$ denotes the integers between $x$ and $y$ inclusive.

$[x,y)$ denotes $[x,y]\setminus\{y\}$.

$\log$ denotes $\log_e$ and $\lg$ denotes $\log_2$.

If $\lambda$ is a security parameter, ``with high probability'' will mean with probability $1-2^{-\epsilon\lambda}$, for  some constant $\epsilon>0$.

Polynomial time or space means time or space polynomial in the security parameter~$\lambda$.

An \emph{arithmetic circuit} $\Phi$ over the ring $\mathsf{R}$ in variables $X = \{x_1,\ldots,x_n\}$  is a directed acyclic graph with every vertex (\emph{gate}) having in-degree either two or zero. Every vertex of in-degree 0 is labelled either by a variable in $X$ or by an element of $\mathsf{R}$. Each other vertex in $\Phi$ has in-degree two and is labelled by either $\times$ or $+$. Every vertex of out-degree 0 in $\Phi$ computes a polynomial in $\mathsf{R}[X]$ in the obvious manner. We refer to the directed edges in the acyclic graph as \emph{wires}. The \emph{depth} of $\Phi$ is the length of the longest directed path in it. (See~\cite{hrubevs2011arithmetic}.)

A \emph{Boolean circuit} $\Phi$ is defined similarly to an arithmetic circuit.
Every vertex of in-degree 0 is labelled either by a variable in $X$ or an element of $\{0,1\}$. Each other vertex in $\Phi$ has in-degree two and is labelled by a binary Boolean function. Then every vertex of out-degree 0 in $\Phi$ computes some Boolean function of the inputs. Note that any finite computation can be represented as a Boolean circuit. (See~\cite{vollmer1999intro}.)

\subsection{Formal Model of Scenario}\label{sec:formalmodel}
We have $n$ integer inputs $m_1, m_2, \ldots, m_n$ distributed in $[0,M)$ according to a probability distribution $\mathcal{D}$. If $X$ is a random integer sampled from $\mathcal{D}$, let $\Pr[X=i]=\xi_i$, for $i\in[0,M)$. We will consider three measures of the \emph{entropy} of $X$, measured in bits:\\[0.5ex]
\begin{tabular}{lr@{}l}
  Shannon entropy & $H_1(X)$ = &\ $-\sum_{i=0}^{M-1}\xi_i \lg \xi_i$,\\[0.5ex]
 Collision entropy & $H_2(X)$ = &\ $-\lg \big(\sum_{i=0}^{M-1}\xi_i^2\big)$,\\[0.5ex]
  Min entropy & $H_\infty(X)$ = &\ $-\lg \big(\max_{i=0}^{M-1}\xi_i\big)$.
\end{tabular}\\[0.5ex]
It is known that $H_1(X)\geq H_2(X)\geq H_\infty(X)$, with equality if and only if $X$ has the uniform distribution on $[0,M)$, in which case all three are $\lg M$. We will denote $H_\infty(X)$ by $\rho$, so it also follows that $H_1(X),H_2(X)\geq\rho$. We use the term ``entropy'' without qualification to mean min entropy, $H_\infty(X)$. Note that  $H_\infty(X)=\rho\geq\lg M$ implies $\xi_i\leq 2^{-\rho}$, $i\in[0,M)$, and that $M\geq 2^\rho$.

We wish to compute a polynomial $P$ of degree $d$ on these inputs. A secure client $A$ selects an instance $\mathcal{E}_K$ of the encryption algorithm $\mathcal{E}$ using the secret parameter set $K$. $A$ encrypts the $n$ inputs by computing $c_i = \mathcal{E}_K(m_i)$, for $i \in [1,n]$. $A$ uploads $c_1, c_2, \ldots, c_n$ and $P'$ to the cloud computing environment, where $P'$ is the homomorphic equivalent of $P$ in the ciphertext space. The cloud computing environment computes $P'(c_1, c_2, \ldots, c_n)$. $A$ retrieves $P'(c_1, c_2, \ldots, c_n)$ from the cloud, and computes \[P(m_1, m_2, \ldots, m_n) = {\mathcal{E}_K}^{-1}(P'(c_1, c_2, \ldots, c_n)).\]

A snooper is only able to inspect $c_1, c_2, \ldots, c_n$, the function $P'$, the computation of $P'(c_1, c_2, \ldots, c_n)$, including subcomputations and working data, and $P'(c_1, c_2, \ldots, c_n)$ itself.

Our encryption schemes are essentially symmetric key encryption, though there is no key distribution problem. The public parameters of our schemes are exposed to the cloud, but they do not provide an encryption oracle.

This model is clearly susceptible to certain attacks. We consider ciphertext only, brute force, and cryptanalytic attacks. To avoid cryptanalytic attacks, we must choose the parameters of the system carefully. Here, a brute force attack will mean guessing the plaintext associated with a ciphertext. In our encryption schemes, it will be true that a guess can be verified. Since $\xi_i\leq 2^{-\rho}$ for $i\in[0,M)$, the expected number $\mu$ of guesses  before making a correct guess satisfies $\mu\geq2^{\rho}$. Massey~\cite{massey1994guessing} gave a corresponding result in terms of the Shannon entropy $H_1(X)$.

It follows similarly that the probability of any correct guess in $2^{\rho/2}$ guesses is at most $2^{-\rho/2}$. This bound holds if we need only to guess  any one of $n$ inputs, $m_1,m_2,\ldots,m_n$, even if these inputs are not independent. Therefore, if $\rho$ is large enough, a brute force  attack is infeasible.

Recall that, in our model, known plaintext attack (KPA) is possible only by brute force, and not through being given a sample of plaintext, ciphertext pairs.

We do not regard chosen plaintext attack (CPA) or chosen ciphertext attack (CCA) as being relevant to our model. Since $\mathcal{E}_K$ is never exposed in the cloud, there is no realistic analogue of an encryption or decryption oracle, as required by these attacks. Of course, in public key encryption, an encryption algorithm is publicly available as part of the system, so CPA must be forestalled. We note that, following \cite{bellare1997concrete}, it is common in studying symmetric key encryption to suppose that defence to CPA or CCA is necessary. While this may provide a stronger notion of security, it seems hard to justify. Both~\cite{bellare2005intro} and~\cite{boneh2015graduate}  provide examples which are intended to justify this convention. However, these examples are unconvincing, and seem to have little practical importance.  Nevertheless, since it is not difficult to do so, we show that the ``N'' variants of our HE schemes below resist CPA.

We note that observation of the function $P'$, which closely resembles $P$, might leak some information about its inputs. However, we assume that this information is far too weak to threaten the security of the system. This assumption seems universal in the existing literature on HE. However, ``garbled circuits''~\cite{bellare2012yao,goldreich1987play} are a possible solution to this problem, if the threat is significant.

Finally, we note that our model of SSCC is very similar to the model of \emph{private single-client computing}, presented in \cite{vandijk2010impossibility} along with an example application.
\subsection{Our Results}
We describe novel practical HE schemes for the encryption of integers, to be employed in a SSCC system inspired by CryptDB \cite{popa2011cryptdb}. CryptDB is an HE scheme where  encryption depends on the operation to be performed. CryptDB encrypts integers using the Paillier cryptosystem \cite{paillier1999} which allows for homomorphic addition. Similar systems (\cite{tetali2013mrcrypt,stephen2014practical}) use Paillier and ElGamal \cite{elgamal1985} to support addition and multiplication, respectively. The ``unpadded'' versions of these schemes are used, which may not be secure under CPA~\cite{goldwasser1984prob}, reducing any possible advantage of a public-key system. However, these schemes do not support both addition and multiplication. To perform an inner product, say, requires re-encrypting the data once the multiplications have been performed so that the additions can then be performed. In a SSCC system, this would require shipping the data back to the initiator for re-encryption, creating a significant overhead. To avoid this problem, we aim for an HE scheme for integers supporting both addition and multiplication.

Our HE scheme over the integers is inspired by the SWHE scheme of van Dijk et al.~\cite{vandijk2010fully} (which we denote DGHVS) that is used as the basis for their public-key system (denoted as DGHV in \cite{coron2011fully}). As in their system, we add multiples of integers to the plaintext to produce a ciphertext. However, DGHVS supports only arithmetic mod~2, and we generalise this to larger moduli.

In the section above, we showed that the input data must have sufficient entropy to negate brute force attacks. If the data lacks sufficient entropy, we will introduce more entropy in two ways. The first is to add random ``noise'' of sufficient entropy to the ciphertext to ``mask'' the plaintext. This approach is employed in DGHV and in sections \ref{sec:he1n} and \ref{sec:he2n}. In our schemes we add a random multiple (from 0 to $\kappa$) of a large integer, $\kappa$, to the ciphertext, such that $m_i<\kappa$, for all $i\in[1,N]$. If the entropy of the original data was $\rho$, once transformed it is $\rho + \lg \kappa$. Therefore, if $\kappa$ is large enough we can ensure that our data has sufficient entropy. However, there is a downside. To prevent the noise term growing so large that the cipherext can no longer be decrypted successfully, we are restricted to computing polynomials of low enough degree.

The other technique will be to increase the dimension of the ciphertext. We represent the ciphertext as a $k$-vector, where each element is a linear function of the plaintext as in DGHVS. Addition and multiplication of ciphertexts are simple transformations of the ciphertexts using vector and matrix algebra. The basic case $k=1$ is described in section \ref{sec:he1}. Then we can increase the entropy $k$-fold by creating a $k$-vector ciphertext. This is because we need to guess $k$ plaintexts to successfully break the system. Assuming that the inputs $m_1,m_2,\ldots,m_n$ are chosen independently from $\mathcal{D}$, and the entropy of inputs is $\rho$, then the entropy of a $k$-tuple $(m_1,m_2,\ldots,m_k)$ is $k\rho$. Thus the $k$-vectors effectively have entropy $k \rho$. If $k$ is chosen large enough, there will be sufficient entropy to prevent brute force attack. Note that the assumption of independence among $m_1,m_2,\ldots,m_n$ can easily be relaxed, to allow some correlation, but we will not discuss the details here. The upside is that some cryptanalytic methods applicable in the case $k=1$ do not seem to generalise even to $k=2$. The downside is that ciphertexts are $k$ times larger, and each homomorphic multiplication requires $\Omega(k^3)$ time and space, in comparison with the case $k=1$. For very large $k$, this probably renders the methods impractical. Therefore, we consider the case $k=2$ in some detail in section~\ref{sec:he2}, before describing the general case in section~\ref{sec:hek}.

Our work here only aims to support integer arithmetic. Other operations, like sorting, require different HE schemes, which we will consider elsewhere. In the integer arithmetic case, a system for computing low-degree polynomials seems to suffice for most practical applications. (See~\cite{naehrig2011can}).  To this end, we will consider practically implementable values for the parameters of the cryptosystems.

\subsection{Related Work}
FHE schemes start by devising a SWHE scheme which supports only homomorphic addition and multiplication, so the computation is to evaluate an arithmetic circuit on encryptions to which random ``noise'' has been edded. A general computation is  represented as an arithmetised Boolean circuit~\cite{babai1991arithmetization}. As ciphertexts are added and multiplied during the computation, the ``noise'' grows until there comes a point where the plaintext cannot be uniquely recovered from the ciphertext.  Therefore the arithmetic circuit must be of limited depth, to prevent this ``noise'' growing too large. If the circuit is sufficiently shallow, this SWHE scheme can be used in its own right, e.g. Boneh et al. \cite{boneh2013pdq}.

The SWHE system is transformed into an FHE scheme by a process of re-encryption where, once the noise grows too large, the ciphertext is re-encrypted, thereby allowing computation to proceed. This re-encryption is performed homomorphically, using a circuit shallow enough that the noise does not grow too large. This circuit necessarily contains information about the private key, which must be suitably hidden. This re-encryption reduces the noise in the ciphertext, and allows circuits of arbitrary depth to be computed homomorphically.

Our scheme is inspired by that of van Dijk et al. \cite{vandijk2010fully}. In their paper they produce an FHE scheme over the integers, where a simple SWHE scheme is ``bootstrapped'' to FHE. van Dijk et al. take a simple symmetric scheme where an plaintext bit $m$ is encrypted as $c = m+2r+pq$, where  the secret key $p$ is an odd $\eta$-bit integer from the interval $[2^{\eta-1},2^\eta)$, and $r$ and $q$ are integers chosen randomly from an interval such that $2r < p/2$. The ciphertext $c$ is decrypted by the calculation $(c \bmod p) \bmod 2$. Our scheme HE1N below may be regarded as a generalisation of this.

van Dijk et al. transform their symmetric scheme into a public key scheme. A public key $\langle x_0,x_1,\ldots,x_\tau \rangle$ is constructed where each $x_i$ is a near multiple of $p$ of the form $pq+r$ where $q$ and $r$ are random integers chosen from a prescribed interval. To encrypt a message a subset $S$ of $x_i$ from the public key are chosen and the ciphertext is calculated as $c=m+2r+2\sum_{i \in S} x_i \mod x_0$. The ciphertext is decrypted as previously described. We could extend our HE$k$N schemes here to a public key variant, using a similar device. However, we will not do so, since public key systems appear to have little application to our model.

van Dijk et al. ``bootstrap'' their public key system  to an FHE scheme, using Gentry's approach~\cite{gentry2009fully}. In this case, the bootstrapping is done by homomorphically simulating a division by $p$, thus obtaining an encryption of $c \bmod p$ which can be used to continue the computation. Our FHE proposal below is based on different principles.

In \cite{coron2011fully}, Coron et al. reduce the size of the public key by using a similar but alternative encryption scheme. In this scheme, $p$ is a prime in the specified interval, $x_0$ is an exact multiple of $p$ and the sum term in the ciphertext is quadratic rather than linear. 

The above FHE schemes 
represent a major theoretical achievement. However, they appear impractical for computations on large data sets, in terms of both running time and storage requirements.

Therefore, the direction of our work is similar to~\cite{naehrig2011can}. The authors implement the SWHE scheme from~\cite{brakerski2011ring}. However, they give results only for degree two polynomials. Our schemes seem capable of computing somewhat higher degree polynomials for practical key and ciphertext sizes.

Recent work on \emph{functional encryption} \cite{goldwasser2013how,goldwasser2013reusable} should also be noted. 
While these results are of great theoretical interest, the scenario where such schemes might be applied is rather different from our model. Also, the methods of~\cite{goldwasser2013how,goldwasser2013reusable} do not seem likely to be of practical interest in the foreseeable future.

The symmetric MORE scheme \cite{kipnis2012efficient} and its derivative \cite{xiao2012efficient} uses linear transformations, as does our scheme HE$k$ in a different way. These systems have been shown~\cite{vizar2015} to be insecure against KPA, at least as originally proposed. However, whether KPA is practically relevant in context is moot.

We also note the work of Cheon et al. \cite{cheon2015crt}. They use the Chinese Remainder Theorem (CRT) in an FHE system. We make use of the CRT in our scheme HE2NCRT below (section \ref{sec:he2ncrt}). However, our construction differs significantly from theirs.

We should note that the encryption of the Boolean circuits in our fully homomorphic system (section \ref{sec:fhe}) has similarities to Yao's ``garbled circuits'' \cite{bellare2012yao,goldreich1987play}.

\subsection{Roadmap}
We present our initial homomorphic scheme in section \ref{sec:inithom} in two variants, HE1 and HE1N. HE1 (section \ref{sec:he1}) is suitable for integers distributed with sufficient entropy. HE1N (section \ref{sec:he1n}) deals with integers not distributed with sufficient entropy, by adding an additional ``noise'' term.

Section \ref{sec:dimension} describes a further two variants, HE2 and HE2N, which increase the entropy of the plaintext by adding a dimension to the ciphertexts, which are 2-vectors. Again, HE2 (section \ref{sec:he2}) deals with integers of sufficient entropy, HE2N (section \ref{sec:he2n}) with integers without the required entropy. We describe this in some detail, since it appears to be practically useful, and is the simplest version of our general scheme.

In section \ref{sec:genk}, we generalise HE2 and HE2N from 2-vectors to $k$-vectors, for arbitrary $k$, in the scheme HE$k$, with noisy variant HE$k$N. These schemes may also be practical for small enough $k$.

In section \ref{sec:he2ncrt}, we present an extension of HE2N, HE2NCRT, which uses the CRT to distribute the computation.

In section \ref{sec:fhe}, we discuss how HE$k$ can be transformed into an FHE scheme for large enough $k$, though the resulting scheme seems only to be of theoretical interest.

In section~\ref{sec:results} we describe extensive experimentation with the schemes, and finally, in section \ref{sec:concfurther}, we give our conclusions.

\section{Initial Homomorphic Scheme}
\label{sec:inithom}
In this section we present details of our initial SWHE schemes over the integers.

\subsection{Sufficient Entropy (HE1)}
\label{sec:he1}
We have integer inputs $m_1, m_2, \ldots, m_n \in [0,M)$. (Negative integers can be handled as in van Dijk et al.~\cite{vandijk2010fully}, by taking residues in $[-(p-1)/2,(p-1)/2)$, rather than $[0,p)$.) We wish to compute a polynomial $P$ of degree $d$ in these inputs. The inputs are distributed with entropy $\rho$, where $\rho$ is large enough, as discussed in section~\ref{sec:formalmodel} above. Our HE scheme is the system $(\Gen,\Enc,\Dec,\Add,\Mult)$.

Let $\lambda$ be a large enough security parameter, measured in bits. Let $p$ and $q$ be suitably large distinct primes such that $p\in[2^{\lambda-1},2^\lambda]$, and $q\in[2^{\eta-1},2^\eta]$, where $\eta\approx\lambda^2/\rho - \lambda$. Here $\lambda$ must be large enough to negate direct factorisation of $pq$ (see~\cite{kleinjung2010factor}), and the relative values of $p$ and $q$ are chosen to negate Coppersmith's attack \cite{coppersmith1997small}.  We will also require $p > (n+1)^dM^d$ to ensure that $P(m_1,m_2,\ldots,m_n) < p$, so that the result of the computation can be successfully decrypted. (In many applications, a smaller value of p may suffice). Our function $\Gen$ will randomly select $p$ and $q$ according to these bounds. Then $p$ is the private symmetric key for the system and $pq$ is the modulus for arithmetic performed by $\Add$ and $\Mult$. $pq$ is a public parameter of the system. We assume that the entropy $\rho\gg\lg\lambda$, so that a brute force attack cannot be carried out in polynomial time.

We can easily set the parameters to practical values. If $n\approx\sqrt{M}$, $M\approx2^\rho$ then we may take $\lambda \approx 3d\rho/2$ and $\eta\approx 3d\lambda/2 - \lambda$ (see appendix~\ref{app:bounds}). For, example, if $\rho=32$, $d=4,$ we can take any $\lambda > 192$, $\eta > 960$.

We encrypt a plaintext integer $m$ as
\begin{align*}
	\Enc(m,p) &= m + r p\ \bmod\ pq
\end{align*}
where $r\uar[1,q)$.

We decrypt the ciphertext $c$ by
\begin{align*}
	\Dec(c,p) &= c\ \bmod\ p
\end{align*}

The sum modulo $pq$ of two ciphertexts, $c = m + rp$ and $c' = m' + r'p$, is
\begin{align*}
	 \Add(c,c')= c+c'\ \bmod\ pq\ =\ m+m' + (r+r')p.
\end{align*}
This decrypts to $m+m'$, provided $m+m'<p$.

The product modulo $pq$ of two ciphertexts, $c = m + rp$ and $c' = m' + r'p$, is
\begin{align*}
	 \Mult(c,c')&= cc' \mod{pq}\\ &= mm' + (rm'+r'm+rr'p)p,
\end{align*}
which decrypts to $mm'$, provided $mm'<p$.

Security of the system is provided by the \emph{partial approximate common divisor problem} (PACDP), first posed by Howgrave-Graham~\cite{howgrave2001approx}, but can be formulated \cite{chen2012faster,cohn2012approx} as:

\begin{definition} \textit{(Partial approximate common divisor problem.)} Suppose we are given one input $x_0=pr_0$ and $n$ inputs $x_i=pr_i + m_i$, $i\in[1,n]$. We have a bound $B$ such that $|m_i| < B$ for all $i$. Under what conditions on the variables, $r_i$ and $m_i$, and the bound $B$, can an algorithm be found that can uniquely determine $p$ in a time which is polynomial in the total bit length of the numbers involved?
\end{definition}

A straightforward attack on this problem is by brute force. Consider $x_1$. Assuming that $m_1$ is sampled from $\mathcal{D}$, having entropy $\rho$, we successively try values for $m_1$ and compute $\gcd(x_0,x_1-m_1)$ in polynomial time until we find a divisor that is large enough to recover $p$. Then we can recover $m_i$ as $(x_i \bmod p)$ for $i\in[2,n]$. As discussed in section~\ref{sec:formalmodel}, the search will requires $2^{\rho}$ $\gcd$ operations in expectation.

Several attempts have been made to solve the PACDP~\cite{howgrave2001approx,cohn2012approx,chen2012faster}, resulting in theoretically faster algorithms for some cases of the problem. However, our parameters for $p$ and $q$ are chosen to negate the attacks of~\cite{howgrave2001approx,cohn2012approx}. The paper~\cite{chen2012faster} gives an algorithm requiring only $\sqrt{M}$ polynomial time operations in the special case that $\mathcal{D}$ is the uniform distribution on $[0,M)$,
and hence $\rho=\lg M$. No algorithm running in time subexponential in $\rho$  is known for this problem in the worst case. Therefore, if $\rho$ is large enough, the encryption should be secure.

In actuality, our system is a special case of PACDP because we use the residues of the approximate prime multiples modulo a distinct semiprime modulus. A semiprime is a natural number that is the product of two prime numbers. A distinct semiprime is a semiprime where the prime factors are distinct. We denote this case of PACDP as the \emph{semiprime partial approximate common divisor problem} (SPACDP). Although it is a restriction, there is no reason to believe that this is any easier than PACDP.

\begin{definition} \textit{(Semiprime factorisation problem.)} Given a semiprime $s$, the product of primes $p$ and $q$, can $p$ and $q$ be determined  in polynomial time?
\end{definition}

The computational complexity of this problem, which lies at the heart of the widely-used RSA cryptosystem,  is open, other than for quantum computing, which currently remains impractical. We will show that breaking HE1 is equivalent to semiprime factorisation. Therefore, our scheme is at least as secure as unpadded RSA~\cite{rivest1978method}.
\begin{restatable}{theorem}{factorise}
\label{thm:1}
	An attack against HE1 is successful in polynomial time if and only if we can factorise a distinct semi-prime in polynomial time.
\end{restatable}

There is a variant of brute force attack on this system, which we will call a \emph{collision attack}. Suppose we have a pair of equal plaintexts $m_1=m_2$. Then the difference between their encryptions $(c_1-c_2)$ is an encryption of $0$, and the scheme is subject to KPA. In fact, if we have $n$ plaintexts $m_1,m_2,\ldots,m_n$, and there exist $i,j\in[1,n]$ with $m_i=m_j$, the product $\Pi_{1\leq i<j\leq n}(c_j-c_i)$ is an encryption of $0$. However, if there is sufficient entropy, this attack is not possible.\vspace{-2ex}
\begin{restatable}{lemma}{collision}
\label{lem:1}
If the inputs $m$ have entropy $\rho$ then,  for any two independent inputs $m_1,m_2$, $\Pr(m_1=m_2)\leq 2^{-\rho}$.	
\end{restatable}

Thus, if we have $n$ inputs, $m_1,m_2,\ldots,m_n$ the probability that there exist $i,j\in[1,n]$ with $m_i=m_j$ is at most $\binom{n}{2} 2^{-\rho}$.	If $n<2^{-\rho/3}$, this probability is at most $2^{-\rho/3}$, smaller than any inverse polynomial in $\lambda$. Hence, for large enough $\lambda$, collision attack is infeasible.

A similar collision attack can be made against the schemes described below. We will not discuss the details, since they are almost identical to those above.

\subsection{Insufficient Entropy (HE1N)}
\label{sec:he1n}
Suppose now that the integer inputs $m_i, i\in[1,n],$ are distributed with entropy $\rho$, where $\rho$ is not large enough to negate a brute force guessing attack. Therefore, we increase the entropy of the plaintext by adding an additional ``noise'' term to the ciphertext. This will be a multiple $s$ (from 0 to $\kappa$) of an integer $\kappa$, chosen so that the entropy $\rho'=\rho + \lg \kappa$ is large enough to negate a brute force guessing attack. We also require $\kappa > (n+1)^d M^d$, so that $P(m_1,m_2,\ldots,m_n) < \kappa$. As a result of the extra linear term in the ciphertext, we compute $P(m_1,\ldots,m_n,\kappa)$ instead. We can easily retrieve $P(m_1,\ldots,m_n)$ from $P(m_1,\ldots,m_n,\kappa)$. $\Gen$ now chooses $p$ and $q$ as in HE1, but with $\eta = \lambda^2/\rho' - \lambda,$ and $p>(n+1)^d (M+\kappa^2)^d$ so that
\[ P(m_1+s_1\kappa,m_2+s_2\kappa,\ldots,m_N+s_n\kappa) < p,\]  when $s_1,s_2,\ldots,s_n\in[0,\kappa)$. The secret key, \sk, is now $(\kappa,p)$.

We can set these parameters to practical values. If we assume $M \approx 2^\rho$ and large enough $n$, as in section \ref{sec:he1}, then we may take $\lg{\kappa} > d (\lg{n} + \rho)$, $\rho'=\rho+\lg\kappa$, $\lambda > d( \lg{n} + 2 \lg{\kappa})$. Then, for example, if $d=3$, $\lg{n}=16$, $\rho=8$, then $\lg{\kappa}>72$,  $\rho'=80$, $\lambda > 480$, $\eta > 2400$. In the extreme case that the inputs are bits, so $\rho=1$,  and $d=3$,  $\lg{n}=16$, then we can take $\lg{\kappa} \approx 51$ and  $\rho'\approx52$, and we have $\lambda > 354$, $\eta > 2056$, which is only 15\% smaller than for $\rho=8$.

We encrypt a plaintext, $m$, as
\[ \Enc(m,\textrm{\sk}) = m+ s\kappa+rp  \mod{pq},\]
where $r\uar[1,q)$ and $s \uar [0,\kappa)$. We decrypt a ciphertext, $c$, as
\[ \Dec(c,\textrm{\sk}) = (c \bmod p) \bmod \kappa. \]
Addition and multiplication of ciphertexts is as above.

The use of random noise gives the encryption the following ``indistinguishability'' property, which implies that the system satisfies IND-CPA \cite{bellare1998relations,bellare1997concrete}.
\begin{restatable}{theorem}{ind}
\label{thm:2}
For any encryption $c$, $c\bmod\kappa$ is polynomial time indistinguishable from the uniform distribution on $[0,\kappa)$. Thus HE1 satisfies IND-CPA, under the assuption that SPACDP is not polynomial time solvable.\vspace{-2ex}
\end{restatable}
\section{Adding a dimension}
\label{sec:dimension}
In this section we discuss adding an additional dimension to the ciphertext, which becomes a 2-vector. In both schemes presented below, HE2 and HE2N, we add a further vector term, with two further secret parameters. The two schemes presented below have a constant factor overhead for arithmetic operations. An addition operation in the plaintext space requires two additions in the ciphertext space, and a multiplication in the plaintext space requires nine multiplications and four additions in the ciphertext space.
\subsection{Sufficient entropy (HE2)}
\label{sec:he2}
As with HE1, it is assumed that the inputs $m_i\ (i\in[1,n])$ are of sufficient entropy. $p$ and $q$ are chosen by $\Gen$ according to the bounds given in section~\ref{sec:he1}. $\Gen$ also sets $\mbf{a}= [a_1\  a_2]^T$, where $a_i\uar [1,pq)$ $(i\in [1,2])$ such that $a_1, a_2, a_1-a_2\neq 0$ ($\bmod\ p$ and $\bmod\ q$). The secret key \sk\ is $(p,\mbf{a})$ and the public parameters are $pq$ and $R$. $R$ is the re-encryption matrix, which is detailed below.

The condition $a_1, a_2, a_1-a_2\neq 0$, $(\bmod~p$,~$\bmod~q)$ fails with exponentially small probability $3(1/p+1/q)$. Thus, $a_1$ and $a_2$ are indistinguishable in polynomial time from $a_1,a_2\uar[0,pq)$.
\subsubsection*{Encryption}
We encrypt a plaintext integer $m$ as the 2-vector $\mbf{c}$,
\begin{align*}
	\mbf{c} =\Enc(m,\textrm{\sk})= (m+rp)\mbf{1} + s\mbf{a} \mod {pq} ,
\end{align*}
where $\mbf{1}= [1\  \,1]^T$, $r\uar [0,q)$, and $s\uar [0,pq)$. We construct $\mbf{c}_{\ta}$, where $c_3=f(c_1,c_2)$ for a given linear function $f$. We will use $f(c_1,c_2)=2c_1-c_2$, though we only require $f(c_1,c_2)\neq c_1,c_2$. Therefore, $c_3= (m+rp) + sa_3 \mod pq$, for $a_3=2a_1-a_2$.

\begin{restatable}{theorem}{hetworandom}
\label{thm:3}
The encryption scheme produces ciphertexts with components which are random integers modulo $pq$.	
\end{restatable}
Note, however, that the components of the ciphertexts are correlated, and this is a vulnerability. We discuss this later in this section (``Cryptanalysis'').

\subsubsection*{Decryption}
To decrypt, we eliminate $s$ from $\mbf{c}$ (modulo $p$), giving
\begin{align*}
	\Dec(\mbf{c},\textrm{\sk})= \mbf{\gamma}^T \mbf{c} \mod p,
\end{align*}
where $\mbf{\gamma}^T= (a_{2}-a_{1})^{-1}[a_2~\,-a_1]$. We call $\gamma$ the \emph{decryption vector}.

\subsubsection*{Addition}
We define the addition operation on ciphertexts as the vector sum modulo $pq$ of the two ciphertext vectors $\mbf{c}$ and $\mbf{c'}$,
\begin{align*}
	\Add(\mbf{c},\mbf{c'})=\mbf{c}+\mbf{c'} \mod{pq}.
\end{align*}

Therefore, if inputs $m,m'$ encrypt as $(m+rp)\mbf{1}+s\mbf{a}$, $(m'+r'p+)\mbf{1}+s'\mbf{a}$, the sum is:
\begin{align*}
	\mbf{c}+\mbf{c'}=(m+m'+(r+r')p)\mbf{1}+(s+s')\mbf{a}.
\end{align*}
which is a valid encryption of $m+m'$.

\subsubsection*{Multiplication}

Consider the Hadamard product modulo $pq$, $\mbf{c}_{\ta} \circ \mbf{c}_{\ta}'$, of the two augmented ciphertext vectors $\mbf{c_{\ta}}$ and $\mbf{c_{\ta}}'$:
\begin{align*}
\mbf{z}_{\ta}=\mbf{c}_{\ta} \circ \mbf{c}_{\ta}' = \begin{bmatrix}
 	c_1 c_1' \\ c_2 c_2'\\c_3c_3' \end{bmatrix} \mod {pq}
\end{align*}

Therefore, if inputs $m,m'$ are encrypted as $(m+rp)\mbf{1}+s\mbf{a}$, $(m'+r'p)\mbf{1}+s'\mbf{a}$, we first calculate
\begin{align*}
\mbf{z}_{\ta} &= (m+rp)(m'+r'p)\mbf{1}_{\ta}+[(m+rp)s'+(m'+r'p)s]\mbf{a}_{\ta}\\
&+ss'\mbf{a}_{\ta}^{\circ2}
=(mm'+r_1p)\mbf{1}_{\ta}+s_1\mbf{a}_{\ta}+ss'\mbf{a}_{\ta}^{\circ2}\ \ \mod {pq},
\end{align*}
where $r_1=mr'+m'r+rr'p$, $s_1=(m+rp)s'+(m'+r'p)s$, and $\mbf{a}_{\ta}^{\circ2}=[a_1^2\ \ a_2^2\ \ a_3^2]^T$.

As we can see, $\mbf{z}_{\ta}$ is not a valid encryption of $mm'$. We need to re-encrypt this product to eliminate the $\mbf{a}_{\ta}^{\circ2}$ term.

We achieve this by multiplying $\mbf{z}_{\ta}$ by $R$, a $2\times3$ matrix, \[\left[\begin{array}{c@{\quad}c@{\quad}c}
         1-2\alpha_1 &        \alpha_{1} & \alpha_1 \\
         -2\alpha_2 &        \alpha_{2}+1 & \alpha_2
\end{array}\right],\]
where $\alpha_1$ and $\alpha_2$ are parameters to be decided.

It is easy to check that $R\mbf1_{\ta}=\mbf1$ and $R\mbf a_{\ta}=\mbf a$, independently of $a_1,a_2$. Now
\begin{align*}
(R\mbf a_{\ta}^{\circ2})_1&=(1-2\alpha_1)a_1^2+\alpha_1a_2^2+\alpha_1(2a_1-a_2)^2\\
&=a_1^2+\alpha_1((2a_1-a_2)^2+a_2^2-2a_1^2)\\
&=a_1^2+2\alpha_1(a_2-a_1)^2 \\
(R\mbf a_{\ta}^{\circ2})_2&=-2\alpha_2a_1^2+(\alpha_2+1)a_2^2+\alpha_2(2a_1-a_2)^2\\
&=a_2^2+\alpha_2((2a_1-a_2)^2+a_2^2-2a_1^2)\\
&=a_2^2+2\alpha_2(a_2-a_1)^2
\end{align*}
Let $\beta=2(a_2-a_1)^2$. Thus, $\beta^{-1}\mod {pq}$ exists. Therefore, if we set
\begin{equation}\label{eq:160}
\alpha_1= \beta^{-1}(\sigma a_1+\varrho p-a_1^2),\qquad\alpha_2= \beta^{-1}(\sigma a_2+\varrho p-a_2^2),
\end{equation}
where $\varrho \uar [0,q]$ and $\sigma \uar [0,pq)$, then we obtain the identity
\[R\mbf a_{\ta}^{\circ2} = \varrho p \mbf{1} + \sigma \mbf{a}.\]

Observe that $\alpha_1,\alpha_2$ are public, but give only two equations for the four parameters of the system $a_1,a_2,\sigma,\varrho p$. These equations are quadratic $\bmod\ pq$, and solving them is as hard as semiprime factorisation in the worst case~\cite{rabin1979dsp}.

We re-encrypt by applying $R$ to $\mbf{z}_{\ta}$, i.e. $\mbf{z}' = R \mbf{z}_{\ta}$, so
\begin{align*}
\mbf{z}'&= (mm'+r_1p)R\mbf{1}+s_1R\mbf{a}+ss'R\mbf{a}^{\circ2}\\
	&= (mm'+r_1p)\mbf{1}+s_1\mbf{a}+ss'(\sigma\mbf{a}+\varrho p\mbf{1})\\
&=(mm'+r_2p)\mbf{1}+(s_1+\sigma rr')\mbf{a}\\
&=(mm'+r_2p)\mbf{1}+s_2\mbf{a}\quad\pmod {pq}
\end{align*}
for some integers $r_2,s_2$. So $\mbf{z}'$ is a valid encryption of $mm'$.

Therefore, the homomorphic equivalent of a multiplication operation is defined as
\begin{align*}
	\Mult(\mbf{c},\mbf{c}') = \mbf{c} \cdot \mbf{c}' = R(\mbf{c_{\ta}}\circ\mbf{c'_{\ta}})\pmod{pq},
\end{align*}
where $\cdot$ is a product on $\mathbb{Z}_{pq}^2$ and $\mbf{c_{\ta}}\circ\mbf{c'_{\ta}}$ is the Hadamard product modulo $pq$ of the two extended ciphertext vectors $\mbf{c_{\ta}}$ and $\mbf{c'_{\ta}}$. Thus, the public parameters of the system are the modulus $pq$ and the re-encryption matrix $R$, i.e. $(pq,R)$.

Observe that, independently of $\mbf a$,
\[R\mbf{c_{\ta}}=(m+rp)R\mbf{1_{\ta}}+sR\mbf{a_{\ta}}=(m+rp)\mbf{1}+s\mbf{a}=\mbf{c},\]
for any ciphertext $\mbf{c}$. Hence re-encrypting a ciphertext gives the identity operation, and discloses no information.

\subsubsection*{Hardness}
We can show that this system is at least as hard as SPACDP. In fact,
\begin{restatable}{theorem}{hardness}
\label{thm:4}
SPACDP is of equivalent complexity to the special case of HE2 where $\delta =a_2-a_1$ ($0<\delta <q$) is known.
\end{restatable}

Observe that, without knowing the parameter $k=a_2-a_1$, HE2 cannot be reduced to SPACDP in this way. Thus HE2 is seemingly more secure than HE1.

\subsubsection*{Cryptanalysis}
Each new ciphertext $\mbf{c}$ introduces two new unknowns $r,s$ and two equations for $c_1,c_2$. Thus we gain no additional information from a new ciphertext. However, if we can guess, $m,\,m'$ for any two ciphertexts $\mbf{c},\mbf{c}'$, we can determine
\begin{align*}
(c_1-m)=rp+sa_1&,\qquad(c_2-m)=rp+sa_2,\\(c'_1-m')=r'p+s'a_1&,\qquad(c'_2-m')=r'p+s'a_2,
\end{align*}
so
\begin{align*}
(c_1-m)&(c'_2-m')-(c_2-m)(c'_1-m')\\=\,&(a_2-a_1)(rs'-r's)p \pmod{pq}
\end{align*}
Since $a_2\neq a_1$, and $sr'\neq s'r$ with high probability, this is a nonzero multiple of $p$, $\nu p$ say. We may assume $\nu<q$, so $p=\gcd(\nu p, pq)$.
We can now solve the linear system $\gamma^T[\mbf{c}\ \,\mbf{c}']=[m\ \,m']\mod{p}$ to recover the decryption vector. This effectively breaks the system, since we can now decrypt an arbitrary ciphertext. We could proceed further, and
attempt to infer $a_1$ and $a_2$, but we will not do so.

Note that to break this system, we need to guess two plaintexts, as opposed to one in HE1. The entropy of a pair $(m,m')$ is $2\rho$, so we have effectively squared the number of guesses needed to break the system relative to HE1. So HE2 can tolerate somewhat smaller entropy than HE1. We note further that HE2 does not seem immediately vulnerable to other attacks on HE1~\cite{howgrave2001approx,cohn2012approx,chen2012faster}.

\subsection{Insufficient entropy (HE2N)}
\label{sec:he2n}
In this section we extend HE1N above (section \ref{sec:he1n}) to two dimensions. $\Gen$ chooses $p,q$ and $\kappa$ according to the bounds given in section \ref{sec:he1n} and $\mbf{1}$, $\mbf{a}$ are as in section \ref{sec:he2}. The secret key is $(\kappa,p,\mbf{a})$, and the public parameters are $pq$ and $R$, as defined in section \ref{sec:he2}.

We encrypt a plaintext integer $m \in [0,M)$ as a 2-vector~$\mbf{c}$,
\begin{align*}
	 \Enc(m,\textrm{\sk}) = \mbf{c} =(m+rp+s\kappa)\mbf{1} + t\mbf{a} \mod {pq},
\end{align*}
where $r$ is as in section~\ref{sec:he2}, $s\uar[0,\kappa)$, and $t\uar [0,pq)$.

We decrypt a ciphertext $\mbf{c}$ by
\begin{align*}
	\Dec(\mbf{c},\textrm{\sk})=(\mbf{\gamma}^T \mbf{c} \mod p)\mod \kappa,
\end{align*}
where $\mbf{\gamma}^T$ is defined as in \ref{sec:he2}.

Addition and multiplication of ciphertexts are defined as in section \ref{sec:he2}.

Finally, we note that HE2N satisfies Theorem~\ref{thm:2}.
\section{Generalisation to \lowercase{\textit{\large k}} dimensions}
\label{sec:genk}
In this section we generalise HE2 and HE2N to $k$-vectors. HE1 and HE1N are the cases for $k=1$ and HE2 and HE2N are the cases for $k=2$.
\subsection{Sufficient entropy (\textbf{\large HE$k$})}
\label{sec:hek}
We now generalise HE2 to $k$ dimensions. $\Gen$, randomly chooses $p$ and $q$ according to the bounds given in section \ref{sec:he2}.  $\Gen$ sets $\mbf{a}_j \uar [1,pq)^k$, $\forall j\in[1,k-1]$. The secret key, {\sk, is $(p,\mbf{a}_1,\ldots,\mbf{a}_{k-1})$, and the public parameters are $pq$ and $R$. Again, $R$ is detailed below.

With regard to computational overhead, the number of arithmetic operations per plaintext multiplication is $O(k^3)$, and the space requirement per ciphertext is $O(k)$, by comparison with HE1.

\subsubsection*{Encryption}
A plaintext, $m \in [0,M]$, is enciphered as
\[\Enc(m,\textrm{\sk})=\mbf{c}= (m+rp)\mbf{1}+\sum_{j=1}^{k-1}s_j\mbf{a}_j \mod{pq} \] where $\mbf{c}$ is a $k$-vector, $r\uar [0,q)$, and $\forall j, s_j \uar [0,pq)$. Let $\mbf{a}_0=\mbf{1}$, and $A_k=[\mbf{a}_0\ \mbf{a}_1\ \ldots\ \mbf{a}_{k-1}]$. We wish the columns of $A_k$ to form a basis for $\mathbb{Z}^k_{pq}$. We will show that they do so with high probability. In the unlikely event that they do not, we generate new vectors until they do.
\begin{restatable}{lemma}{lemten}
\label{lem:10}
$\Pr(\mbf{a}_0,\mbf{a}_1,\ldots,\mbf{a}_{k-1}$ do not form a basis$)\leq(k-1)(1/p+1/q)$.
\end{restatable}
We extend our definition of an augmented vector, $\mbf{v}_\star$, for a $k$-vector, $\mbf{v}$, such that $\mbf{v}_\star$ is a $\binom{k+1}{2}$-vector, with components $v_i$ ($1\leq i \leq k$) followed by $2v_i-v_j$ ($1\leq i<j\leq k$). 
In general, for $\ell>k$, $v_\ell=2v_i-v_j$, where $\ell=\binom{i}{2}+k+j-1$. Note that $\mbf{v}_\star=U_k\mbf{v}$ for a $\binom{k+1}{2}\times k$ matrix with entries $0,\pm 1,2$, and whose first $k$ rows form the $k\times k$ identity matrix $I_k$. 
Note that $\mbf{v}_\star=U_k\mbf{v}$ implies that $\mbf{1}_\star$ is the $\binom{k+1}{2}$ vector of 1's, and that $*$ is a linear mapping, i.e. $(r_1\mbf{v}_1+r_2\mbf{v}_2)_\star=r_1\mbf{v}_{1*}+r_2\mbf{v}_{2*}$.
\subsubsection*{Decryption}
\[\Dec(\mbf{c},\textrm{\sk})= \mbf{\gamma}^T\thsp\mbf{c}\mod p.\]
where $\mbf{\gamma}^T=(A_k^{-1})_1$ is the first row of $A_k^{-1}$. We call $\mbf{\gamma}$ the \emph{decryption vector}, as in HE2.
\subsubsection*{Addition}
Addition of ciphertexts is the vector sum of the ciphertext vectors as with HE2.

\subsubsection*{Multiplication}
Consider the Hadamard product of two augmented ciphertext vectors, $\mbf{c}_\star\circ\mbf{c}'_\star$. For notational brevity, let $\tilde{m}=m+rp$.
\begin{align*}
\mbf{c}_\star\circ\mbf{c}'_\star\ &=\ \big(\tilde{m}\mbf{1}_\star+\sum_{j=1}^{k-1}s_j\mbf{a}_{\star j}\big)\circ\big(\tilde{m}'\mbf{1}_\star+\sum_{j=1}^{k-1}s'_j\mbf{a}_{\star j}\big)\\
 &=\ \tilde{m}\tilde{m}'\mbf{1}_\star+\sum_{j=1}^{k-1}(\tilde{m}s'_j+\tilde{m}'s_j)\mbf{a}_{\star j}
 +\sum_{j=1}^{k-1}s_js'_j\mbf{a}_{\star j}\circ\mbf{a}_{\star j}\\ &+
\sum_{1\leq i< j\leq k-1}(s_is'_j+s_i's_j)\mbf{a}_{\star i}\circ\mbf{a}_{\star j},
\end{align*}
since $\mbf{1}_\star\circ\mbf{v}_\star=\mbf{v}_\star$ for any $\mbf{v}$. There are $\binom{k}{2}$ product vectors, which we must eliminate using the re-encryption matrix, $R$.

The re-encryption matrix, $R$, is $k\times\binom{k+1}{2}$. We require that $R\mbf{v}_\star=\mbf{v}$, for all $\mbf{v}$.
\begin{restatable}{lemma}{lemtwenty}
\label{lem:20}
Let $A_{\star k}=[\mbf{a}_{\star 0}\ \mbf{a}_{\star 1}\ \ldots\ \mbf{a}_{\star ,k-1}]$, where the columns of $A_k$ form a basis for $\mathbb{Z}^k_{pq}$. If $RA_{\star k}=A_k$, then $R\mbf{v}_\star=\mbf{v}$ for all $\mbf{v}\in\mathbb{Z}^k_{pq}$.	
\end{restatable}
The condition $RA_{\star k}=A_k$ can be written more simply, since it is $RU_kA_k=A_k$. Postmultiplying by $A_k^{-1}$ gives
$RU_k=I_k$. 
%

Since $RA_{\star k}=A_k$, we have
\begin{align*}
R(\mbf{c}_\star\circ\mbf{c}'_\star)\
 &=\ (mm'+\hat{r}p)\mbf{1}+\sum_{j=1}^{k-1}\hat{s}_j\mbf{a}_{j}\\ &\hspace*{1cm}+
\sum_{1\leq i\leq j\leq k-1}\hat{s}_{ij}R(\mbf{a}_{\star i}\circ\mbf{a}_{\star j}),
\end{align*}
where $\hat{r}$, $\hat{s}_j$  and $\hat{s}_{ij}$ ($1\leq i<j\leq k-1$) are some integers.

There are $k(\binom{k+1}{2}-k)=k\binom{k}{2}$ undetermined parameters $R_{i\ell}$, $1\leq i\leq k$, $k < \ell \leq \binom{k+1}{2}$. We now determine these by setting
\begin{equation}\label{equ:10}
 R(\mbf{a}_{\star i}\circ\mbf{a}_{\star j})\ =\ \varrho_{ij}p\thsp\mbf{1}+\sum_{l=1}^{k-1}\sigma_{ijl}\thsp\mbf{a}_l
\end{equation}
Thus we have $k\binom{k}{2}$ new unknowns, the $\varrho$'s and $\sigma$'s, and $k\binom{k}{2}$ linear equations for the $k\binom{k}{2}$ unassigned $R_{i\ell}$'s.
Let  $A^{\circ2}_{\star k}$ be the $\binom{k+1}{2}\times\binom{k+1}{2}$ matrix with columns $\mbf{a}_{\star i}\circ\mbf{a}_{\star j}$ ($0\leq i< j < k$), and let $C_k$ be the $k\times\binom{k}{2}$ matrix with columns $\varrho_{ij}p\thsp\mbf{1}+\sum_{l=1}^{k-1}\sigma_{ijl}\thsp\mbf{a}_l$ ($0<i<j<k$). Then the equations for the $R_{i\ell}$ can be written as
\begin{equation}\label{equ:20}
RA^{\circ2}_{\star k}\ =\ \left[A_k \mid C_k\right].
\end{equation}
giving $k\binom{k+1}{2}$ linear equations for the $k\binom{k+1}{2}$ $R_{i\ell}$'s in terms of quadratic functions of the  $k(k-1)$ $a_{ij}$'s ($1\leq i\leq k, 1\leq j\leq k-1$), which are undetermined. Thus the system has $k(k-1)$ parameters that cannot be deduced from $R$.

The system of equations~\eqref{equ:20} has a solution provided that $A^{\circ2}_{\star k}$ has an inverse $\bmod\ pq$. We prove that this is true with high probability. Again, in the unlikely event that this is not true, we generate new vectors $\mbf{a}_1,\ldots,\mbf{a}_{k-1}$ until it is.
\begin{restatable}{theorem}{thmten}
\label{thm:10}
$A^{\circ2}_{\star k}\mbox{ has no inverse\,} \bmod{pq}$ with probability at most $(k^2-1)(1/p+1/q)$.	
\end{restatable}
Note that Theorem~\ref{thm:10} subsumes Lemma~\ref{lem:10}, since the first $k$ columns of $A^{\circ2}_{\star k}$ contain $A_k$ as a submatrix, and must be linearly independent.

 Each $\mbf{c}$ introduces $k$ new parameters
$rp,s_1,\ldots,s_{k-1}$ and $k$ equations, so the number of undetermined parameters is always $k(k-1)$.

\subsubsection*{Cryptanalysis}\label{HEk cryptanalysis}
Note that $p$ can still be determined if we know $m_i$ for $k$ ciphertexts. Then let
\[ C=[\mbf{c}_1-m_1\mbf{1}\ \ldots\ \mbf{c}_k-m_k\mbf{1}],\quad A_k=[\mbf{1}\ \mbf{a}_1\ \ldots\ \mbf{a}_{k-1}]\]
and let
\[ W=\left[\begin{array}{c@{\ \ }c@{\ \ }c@{\ \ }c@{\ \ }c@{\ \ }c}
         r_{1}p & r_{2}p  & \ldots & r_{k}p \\
         s_{1,1} & s_{2,1} & \ldots & s_{k,1} \\
         \vdots & & & \vdots \\
          s_{1,k-1} & s_{2,k-1} & \ldots & s_{k,k-1}
\end{array}\right],\]
\[W'=\left[\begin{array}{c@{\ \ }c@{\ \ }c@{\ \ }c@{\ \ }c@{\ \ }c}
         r_{1} & r_{2}  & \ldots & r_{k} \\
         s_{1,1} & s_{2,1} & \ldots & s_{k,1} \\
         \vdots & & & \vdots \\
          s_{1,k-1} & s_{2,k-1} & \ldots & s_{k,k-1}
\end{array}\right],\]
where $r_i,s_{ij}$ refer to $\mbf{c}_i$. Then $C=A_kW$, and so $\det C=\det A_k\det W$. Note that $\det W=p\det W'$, so $\det C$ is a multiple of $p$. Now $\det C$ can be determined in $O(k^3)$ time and, if it is nonzero, $p$ can be determined as $\gcd(\det C,pq)$. Then $p$ can be recovered if $\det C\neq 0$.
\begin{restatable}{lemma}{lemthirty}
\label{lem:30}
$\Pr(\det C = 0 \bmod\, pq)\leq (2k-1)(1/p+1/q)$.	
\end{restatable}
Once we have recovered $p$, we can use the known $m_i$ to determine the decryption vector $\mbf{\gamma}$, by solving linear equations. Let \[C_0 = [\mbf{c}_1\ \mbf{c}_2\ \ldots\ \mbf{c}_k],\quad \mbf{m}^T = [m_1\ m_2\ \ldots\ m_k].\vspace{-2ex}\]
\begin{restatable}{lemma}{lemthirtyfive}
\label{lem:35}
$\Pr(\det C_0 = 0 \bmod\, pq)\leq (2k-1)(1/p+1/q)$.
\end{restatable}
Thus, with high probability, we can solve the system \[ \mbf{\gamma}^TC_0= \mbf{m}^T\quad\mod{p}\] uniquely, to recover $\mbf{\gamma}$ and enable decryption of an arbitrary ciphertext. However, encryption of messages is not possible, since we gain little information about $\mbf{a}_1,\ldots,\mbf{a}_k$. Note also that, if we determined $p$ by some means other  than using $k$ known plaintexts, it is not clear how to recover $\mbf{\gamma}$.

To break this system, we need to guess $k$ plaintexts. The entropy of a $k$-tuple of plaintexts $(m_1,m_2,\ldots,m_k)$ is $k\rho$, so effectively we need $\mu^k$ guesses, where $\mu$ is the number of guesses needed to break HE1. So HE$k$ can tolerate much smaller entropy than HE1, provided $k$ is large enough. If $k$ is sufficiently large, the scheme appears secure without adding noise, though it does not have the other advantages of adding noise. We discuss this further in section~\ref{sec:fhe}.

\subsubsection*{Fixing an insecurity for $k>2$}
The decryption vector for HE$k$ is $\mbf{\gamma}^T=(A_k^{-1})_1$.  Note that $\mbf{\gamma}^T\mbf{1}=1$  and $\mbf{\gamma}^T\mbf{a}_i=0$ ($i\in[1,k-1]$), since $\mbf{\gamma}^T\mbf{a}_i=I_{1i}$ ($i\in[0,k-1]$).

The equations
\begin{equation}\label{equ:100}
 R(\mbf{a}_{\star i}\circ\mbf{a}_{\star j})\ =\ p\varrho_{ij}\thsp\mbf{1}+\sum_{l=1}^{k-1}\sigma_{ijl}\thsp\mbf{a}_l,
\end{equation}
define a product $\cdot$ on $\mathbb{Z}^k_{pq}$ so that
$\mbf{c}\cdot\mbf{c}'=R(\mbf{c}_{\star}\circ\mbf{c}'_{\star})$. This product is linear, commutative and  distributive, since $R$ and $\star$ are linear operators, and $\circ$ is commutative and  distributive. So we have an algebra $\mcl{A}_k$, with unit element $\mbf{1}$ \cite{schafer1966introduction}. The $\varrho_{ij},\sigma_{ijl}$ ($i,j,l\in[1,k-1])$ are the \emph{structure constants} of the algebra. In general, $\mcl{A}_k$ will not be associative, i.e. we can have
\begin{align*}
	R(R(\mbf{c}_{1\star}\circ\mbf{c}_{2\star})_*\circ\mbf{c}_{3\star})&=(\mbf{c}_1\cdot\mbf{c}_2)\cdot\mbf{c}_3\\
\neq\mbf{c}_1\cdot(\mbf{c}_2\cdot\mbf{c}_3)&=R(\mbf{c}_{1\star}\circ R(\mbf{c}_{2\star}\circ\mbf{c}_{3\star})_*).
\end{align*}
This leads to the following potential insecurity. We must have
\begin{equation}\label{eq:110}
\mbf{\gamma}^T((\mbf{c}_1\cdot\mbf{c}_2)\cdot\mbf{c}_3)\ =\ \mbf{\gamma}^T(\mbf{c}_1\cdot(\mbf{c}_2\cdot\mbf{c}_3))\quad\pmod p,
\end{equation}
in order to have correct decryption. The \emph{associator} for $\mcl{A}_k$~is
\begin{align*}
[\mbf{c}_i,\mbf{c}_j,\mbf{c}_l]\ &= \mbf{c}_i\cdot(\mbf{c}_j\cdot\mbf{c}_l)-(\mbf{c}_i\cdot\mbf{c}_j)\cdot\mbf{c}_l\\ &=rp\mbf{1}+\sum_{l=1}^{k-1}s_{l}\thsp\mbf{c}_l\ \, \pmod {pq}.
\end{align*}
Thus $[\mbf{c}_i,\mbf{c}_j,\mbf{c}_l]$ is an encryption of $0$. If we can find $k$ such associators from $\mbf{c}_1,\ldots,\mbf{c}_n$ which violate~\eqref{eq:110}, then with high probability we will have $k$ linearly independent associators.   We can use use these to make a collision attack on HE$k$, in a similar way to that described in section~\ref{sec:he1}. We use the $\gcd$ method to determine $p$, and then $\mbf\gamma$, as described in section~\ref{HEk cryptanalysis}. In fact all we need is that~\eqref{eq:110} holds for any associator. That is, for all $\mbf{c}_1, \mbf{c}_2, \mbf{c}_3$, we need
\begin{equation*}
\mbf{\gamma}^T((\mbf{c}_1\cdot\mbf{c}_2)\cdot\mbf{c}_3)\ =\ \mbf{\gamma}^T(\mbf{c}_1(\cdot\mbf{c}_2\cdot\mbf{c}_3))\quad\pmod {pq},
\end{equation*}
or, equivalently, using the CRT,
\begin{equation}\label{eq:120}
\mbf{\gamma}^T((\mbf{c}_1\cdot\mbf{c}_2)\cdot\mbf{c}_3)\ =\ \mbf{\gamma}^T(\mbf{c}_1\cdot(\mbf{c}_2\cdot\mbf{c}_3))\quad\pmod {q}.
\end{equation}
By linearity, it follows that~\eqref{eq:120}  holds if and only if it holds for all basis elements, excluding the identity. That is, for all $i,j,l\in[1,k-1]$, we need
\begin{equation}\label{eq:140}
\mbf{\gamma}^T(\mbf{a}_i\cdot(\mbf{a}_j\cdot\mbf{a}_l))\ =\ \mbf{\gamma}^T((\mbf{a}_i\cdot\mbf{a}_j)\cdot\mbf{a}_l)\quad\pmod {q}.
\end{equation}
The associator for $\mcl{A}_k$ is
\begin{align*}
[\mbf{a}_i,\mbf{a}_j,\mbf{a}_l]\ &= \mbf{a}_i\cdot(\mbf{a}_j\cdot\mbf{a}_l)-(\mbf{a}_i\cdot\mbf{a}_j)\cdot\mbf{a}_l\\ &=rp\mbf{1}+\sum_{l=1}^{k-1}s_{l}\thsp\mbf{a}_l\ \, \pmod {pq},
\end{align*}
for some integers $r,s_1,\ldots,s_{k-1}$, and so $\mbf{\gamma}^T[\mbf{a}_i,\mbf{a}_j,\mbf{a}_l]=rp$.

If $\mcl{A}_k$ is associative, the problem does not arise, since \eqref{eq:140} will be satisfied automatically. Associativity holds for $k\leq2$, since all we have to check is that $\mbf{a}\cdot(\mbf{a}\cdot\mbf{a})=(\mbf{a}\cdot\mbf{a})\cdot\mbf{a}$, which is true by commutativity. Thus HE$k$ with $k\leq2$ cannot be attacked in this way.
%

Requiring associativity in $\mcl{A}_k$ would overconstrain the system, since it imposes $k\binom{k+1}{2}$ equations on the $k\binom{k+1}{2}$ structure constants. We have only $k(k-1)$ undetermined parameters, so this is too much. But all we need is that \eqref{eq:140} holds. We have the following.

\begin{restatable}{lemma}{lemoneforty}
\label{lem:140}
 \eqref{eq:140} holds if and only if
\[ \ts\sum_{t=1}^{k-1}\sigma_{jlt}\varrho_{it}=\ts\sum_{t=1}^{k-1}\sigma_{ijt}\varrho_{lt}\pmod{q},\ \forall i,j,l\in[1,k-1].\]	
\end{restatable}
There are several ways to ensure that~\eqref{eq:140} holds. We will do so by giving the $\varrho_{ij}$ a multiplicative structure.
\begin{restatable}{lemma}{lemonefifty}
\label{lem:150}
Let $\tau,\varrho_i\uar[0,q)$ $(i\in[1,k-1]$), let $\varrho_{ij}=\varrho_i\varrho_j \mod{q}$, and let the $\sigma_{ijl}$ satisfy $\sum_{l=1}^{k-1}\sigma_{ijl}\thsp\varrho_l=\tau\varrho_i\varrho_j\pmod{q}$ for all $i,j\in[1,k-1]$. Then, for all $i,j,\ell\in[1,k-1]$, $\mbf{\gamma}^T(\mbf{a}_i\cdot(\mbf{a}_j\cdot\mbf{a}_l))=\tau\varrho_i\varrho_j\varrho_l\mod{q}$, the symmetry of which implies that~\eqref{eq:140} holds.	
\end{restatable}
Thus the conditions of Lemma~\ref{lem:150} are sufficient to remove the insecurity. The price is that we now have $(k-1)\binom{k}{2}+(k-1)+k(k-1)=(k+1)\binom{k}{2}+k-1$ parameters and $k\binom{k}{2}$ equations. There are $\binom{k}{2}+(k-1)=(k+2)(k-1)/2$ independent parameters. This is fewer than the original $k(k-1)$, but remains $\Omega(k^2)$.

\subsection{Insufficient entropy (\textbf{\large HE$k$N})}
\label{sec:hekn}
In this section, we generalise HE2N to $k$ dimensions. $\Gen$, randomly chooses $\kappa$, $p$ and $q$ according to the bounds given in section \ref{sec:he2n}. $\forall j$, $\Gen$ sets $\mbf{a}_j$ as in \ref{sec:hek}. The secret key, \sk, is ($\kappa$, $p$, $\mbf{a}_1$, $\ldots$, $\mbf{a}_{k-1}$), and the public parameters are $pq$ and $R$. $R$ is as given in section \ref{sec:hek}. Note that, as a result of adding the ``noise'' term, defence against non-associativity is not required.

A plaintext, $m \in [0,M]$, is enciphered as
\[\Enc(m,\textrm{\sk})=\mbf{c}= (m+rp+s\kappa)\mbf{1}+\sum_{j=1}^{k-1}t_j\mbf{a}_j \pmod{pq} \] where $r,s$ are as in section \ref{sec:he2n}, and $t_j \uar [0,pq)$ $\forall j\in[1,k)$.

A ciphertext is deciphered by,
\[\Dec(\mbf{c},\textrm{\sk})= (\mbf{\gamma}^T\thsp\mbf{c}\mod p) \mod \kappa.\]
where $\mbf{\gamma}^T$ is defined as in section \ref{sec:hek}.

Addition and multiplication of ciphertexts are as in section~\ref{sec:hek}.

The effective entropy of HE$k$N is $\rho'=k(\rho + \lg \kappa)$. Thus, as we increase $k$, the ``noise'' term can be made smaller while still providing the requisite level of entropy.

Clearly HE$k$N also inherits the conclusions of Theorem~\ref{thm:2}.
\section{An extension of HE2N using the CRT (HE2NCRT)}
\label{sec:he2ncrt}
As an interesting aside, we extend HE2N (section \ref{sec:he2n}) using a technique inspired by CRT secret sharing, so that we compute the final result modulo a product of primes $\prod_{j=1}^K p_j$ rather than modulo $p$, where $K$ is the number of primes.

In this scheme, we distribute the computation. We have $K$ processors. Each processor computes arithmetic on ciphertexts modulo $p_jq_j$, where $p_j,q_j$ are suitable primes. Also, each processor only receives the $j$th ciphertext vector of an integer. Addition and multiplication of ciphertexts is as defined in section \ref{sec:he2}, except that it is performed modulo $p_jq_j$.

This serves two purposes. The first is to be able to handle larger moduli by dividing the computation into subcomputations on smaller moduli. The second is to mitigate against exposure of the secret key $p$ in the system presented in section \ref{sec:he2n}, by not distributing the modulus $pq$ to each processor. Instead, we distribute $p_jq_j$ to the $j$th processor, for $j \in [1,K]$. This allows us to partition the computation into subcomputations, encrypted using different parameters. Thus, should an attacker compromise one subcomputation, they may gain no knowledge of other subcomputations.
\subsubsection*{Key Generation}
$\Gen$, randomly chooses $\kappa$ as in section \ref{sec:he1n}. For all $j \in [1,K]$, it randomly chooses a prime $p_j$ such that $p_j$ satisfies $2^{\lambda-1}<p_j<2^\lambda$ and \[\Pi = \prod\limits_{j=1}^K p_j > (n+1)^d (M+\kappa^2)^d.\]
It also randomly chooses $q_j$,  $j \in [1,K]$, as for $q$ in section \ref{sec:he1}. Finally, it sets $\mbf{a}_j = [a_{j1}\  a_{j2}]^T$, where $a_{jk}\uar [1,p_jq_j)$ $( j \in[1,K], k\in [1,2])$ such that $a_{j1}\neq a_{j2}\mod p$ and $a_{j1}\neq a_{j2}\mod q$.
The secret key, \sk, is $(\kappa,p_1,\ldots,p_K,\mbf{a}_1,\ldots,\mbf{a}_K)$, and the public parameters are $p_jq_j$ $(j\in[1,K])$ and $R_j$ $(j\in[1,K])$, where each $R_j$ is defined as $R$ in section \ref{sec:he2}.
\subsubsection*{Encryption}
We encrypt an integer, $m_i$ ($i\in[1,n]$), as the set of $K$ 2-vectors, $\mbf{c}_{ij}$,
\begin{align*}
	\mbf{c}_{ij} = (m_i + r_{ij} p_j + s_i\kappa)\mbf{1} + t_{ij} \mbf{a}_j \bmod p_j q_j \  (j \in [1,K]),
\end{align*}
where $r_{ij}\uar [0,q_j)$, $s_i\uar[0,\kappa)$, and $t_{ij} \uar$ $[0,p_jq_j)$ $(i\in[1,n], j \in [1,K])$.

\subsubsection*{Decryption}
To decrypt, we first decrypt the $j$th ciphertext of the computational result $\mbf{c}_j$ as in section~\ref{sec:he2}, to give
\begin{align*}
	P_j= \mbf{\gamma}_j^T \mbf{c}_j \mod {p_j},
\end{align*}
where $P_j$ is the residue of $P(m_1,m_2,\ldots,m_n,\kappa) \mod p_j$ and $\mbf{\gamma}_j^T=(a_{j2}-a_{j1})^{-1}[a_{j2}\ -a_{j1}]$.

We then use the Chinese Remainder Theorem to compute the plaintext as \[P(m_1,m_2,\dots,m_n)=\bigg(\sum\limits_{j=1}^K P_jM_j \mu_j  \mod \Pi\bigg)\mod \kappa,\] where $M_j = \Pi/p_j$ and $\mu_j=M_j^{-1} \bmod{p_j}$.

\subsubsection*{Addition and Multiplication}
Addition of ciphertexts is performed as in \ref{sec:he2}. Multiplication of ciphertexts on processor $j$ is now \[ \Mult(c_j,c'_j)=R_j (c_{j \star} \circ c'_{j \star}).\]

Clearly HE$k$N could be extended to HE$k$NCRT in a similar way, but we will not discuss the details here.

\section{Fully Homomorphic System}
\label{sec:fhe}
We return to HE$k$, presented above in section \ref{sec:hek}. We will show that, for large enough $k$, this can be used to create an FHE system.

We may use HE$k$ to evaluate an arithmetic circuit homomorphically, where $\mathsf{R}=\mathbb{Z}_{pq}$. However, this system is only somewhat homomorphic. If the computational result grows larger than $p$, we are unable to successfully decrypt it. This restricts us to arithmetic circuits of bounded depth to avoid the blow up. To make it fully homomorphic, we consider Boolean circuits.

Typically, we will use  the binary Boolean functions, AND, OR, and NOT in the Boolean circuit. However, we may use fewer  functions. Any Boolean circuit may be represented using only NAND gates \cite{scharle1965}. Recall that the indegree of any gate in the circuit is always $2$, but the outdegree is arbitrary.  The inputs to each gate are bits, as are the outputs. We will denote the set of inputs to the circuit by $I\subseteq V$, and the set of outputs by $O\subseteq V$.  The inputs have indegree 0, and the outputs have outdegree 0, but we regard the inputs as having indegree 1, and the outputs as having outdegree 1, with wires from and to the external environment $\Lambda$.

Note that constant input bits can easily be eliminated from the circuit, so we assume there are none, to avoid an obvious KPA. Even so, if we represent the bit values $0,1$ by encrypting known constants $\alpha_0,\,\alpha_1$, the HE$k$ system is open to a simple KPA. For any ciphertext $\mbf{c}$, we can compute $\mbf{c}'=(\mbf{c}-\alpha_0\mbf{1})\cdot(\mbf{c}-\alpha_1\mbf{1})$.
Then $\mbf{c}'$ is an encryption of~$0$. By repeating this on $k$ ciphertexts, we can obtain $k$ linearly independent encryptions of zero with high probability. Once we have done this, we can determine $p$ and $\mbf\gamma$ as in section~\ref{HEk cryptanalysis}. The problem, of course, is that we have not increased the entropy of the input data.

Therefore, we must add noise to the ciphertexts, but we will do this so as to ensure that the noise does not grow with the the depth of the circuit. On each wire $e\in E$, we will represent the bit value $b_e\in\{0, 1\}$ by $w_e\in\{\omega_{0e},\omega_{1e}\}$, where $\omega_{0e}=2s_{0,e}$, $\omega_{1e}=1+2s_{1,e}$, where $s_{0,e},s_{0,e}\uar[0,\kappa)$. Thus $b_e= w_e \bmod\ 2$, and the noise has entropy $\lg\kappa$. The value of $\kappa$ is chosen as large as possible such that we can correctly evaluate any polynomial of degree 2 in two variables.  For each input $i\in I$, we represent the input bit value $b_i$ similarly, by $x_i\in\{\omega_{0i},\omega_{1i}\}$. The inputs and the wires in the circuit are encrypted using HE$k$, the inputs directly and the other wires indirectly as described below. As discussed in section \ref{sec:hek}, we need $k$ known plaintexts to break HE$k$. The plaintexts are the encrypted bits $w_e \bmod\ 2$, so a brute force attack requires guessing at least $2^k$ bits. So, by setting $k\gg \log \lambda$, a brute force attack on the system requires time superpolynomial in the security parameter $\lambda$.


An input $i\in I$ has a wire $(\Lambda,i)$ on which the (encrypted) input value $x_i$ is stored. For any wire $e=(i,v)$ from input $i$, we have a linear function $L(x)=a+bx$, which converts the plaintext input value $x\in\{\alpha_0,\alpha_1\}$ to the wire value $w\in\{\gamma_0,\gamma_1\}$. (We suppress the wire labels $e$ when they are clear from the context.) This requires
\[ a=(\alpha_1-\alpha_0)^{-1}(\alpha_1\gamma_0-\alpha_0\gamma_1),\quad b=(\alpha_1-\alpha_0)^{-1}(\gamma_1-\gamma_0).\]
The encrypted coefficients of this function are stored as data for the wire $e$. 
Note that all computations are $\bmod\ pq$, and the required inverses exist because the numbers involved are less than $\kappa$.


For each output wire $e=(v,v')$ of a NAND gate $v$, we have a quadratic function $Q(x,y)=a+bx+cy+dxy$, which converts the values on the input wires of the gate, $x\in\{\alpha_0,\alpha_1\}$, $y\in\{\beta_0,\beta_1\}$, to the wire value $w\in\{\gamma_0,\gamma_1\}$. This requires
\begin{align*}
a=\gamma_0+\alpha_1\beta_1\vartheta,\ \ b=-\beta_1\vartheta,\ \ c = -\alpha_1\vartheta,\ \ d = \vartheta,
\end{align*}
where $\vartheta=\big((\alpha_1-\alpha_0)(\beta_1-\beta_0)\big)^{-1}(\gamma_1-\gamma_0)$. Again, the encrypted coefficients of this function are stored as data for the wire $e$. 

For each output gate $v\in O$, we decrypt the value $w\in \{\gamma_0,\gamma_1\}$ computed by its (unique) output wire $(v,\Lambda)$. Then the output bit is $w \bmod\, 2$.

Thus we replace the logical operations of the Boolean circuit by evaluation of low degree polynomials. For simplicity, we have chosen to use only NAND gates, but we can represent any binary Boolean function by a quadratic polynomial in the way described above. Since the quadratic polynomials are encrypted in our system, they conceal the binary Boolean function they represent. Thus the circuit can be ``garbled''~\cite{bellare2012yao,goldreich1987play}, to minimise inference about the inputs and outputs of the circuit from its structure.

However, there is a price to be paid for controlling the noise. The encrypted circuit is not securely reusable with the same values $\omega_{0e},\omega_{1e}$ for $w_e$. Suppose we can observe the encrypted value on wire $e$ three times giving cyphertexts $\mbf{c}_1,\mbf{c}_2,\mbf{c}_3$. Two of these are encryptions of the same value $2s_{0,e}$ or $1+2s_{1,e}$. Thus $(\mbf{c}_1-\mbf{c}_2)\cdot(\mbf{c}_1-\mbf{c}_3)\cdot(\mbf{c}_2-\mbf{c}_3)$ is an encryption of $0$.
By doing this for $k$ wires, we can break the system. This is essentially the collision attack described in section~\ref{sec:inithom}.

Some reuse of the encrypted circuit is possible by using multiple values on the wires, and higher degree polynomials for the gates. However, we will not consider this refinement, since the idea seems to have little practical interest.

\section{Experimental Results}
\label{sec:results}

\begin{table*}[!thp]
\centering
\begin{adjustbox}{width=\textwidth,totalheight=8in,keepaspectratio}
\begin{tabular}{llllllllll}
\toprule
Alg.&  \multicolumn{3}{c}{Parameters} &  \multicolumn{2}{c}{Encryption} &  \multicolumn{3}{c}{MR Job} & Decrypt(ms) \\
 & $d$ & $\rho$ & $\rho'$ & Init(s) & Enc($\mu$s) & Exec(s) & Prod($\mu$s) & Sum($\mu$s) &  \\
\midrule
HE1       & 2 & 32  & n/a      & 0.12              & 13.52              & 23.82               & 54.41              & 9.06           & 0.21            \\
HE1       & 2 & 64  & n/a      & 0.12              & 16.24              & 23.85               & 60.38              & 8.04           & 0.49            \\
HE1       & 2 & 128 & n/a      & 0.15              & 25.73              & 23.77               & 84.69              & 8.43           & 0.28            \\
HE1       & 3 & 32  & n/a      & 0.17              & 22.98              & 23.65               & 87.75              & 11.46          & 0.35            \\
HE1       & 3 & 64  & n/a      & 0.19              & 34.63              & 24.72               & 95.68              & 12.37          & 0.45            \\
HE1       & 3 & 128 & n/a      & 0.42              & 54.83              & 26.05               & 196.71             & 14.07          & 0.55            \\
HE1       & 4 & 32  & n/a      & 0.28              & 43.36              & 24.48               & 108.72             & 13.75          & 0.5             \\
HE1       & 4 & 64  & n/a      & 0.53              & 58.85              & 26.41               & 227.44             & 15.85          & 3.59            \\
HE1       & 4 & 128 & n/a      & 1.36              & 104.95             & 28.33               & 484.95             & 16.92          & 5.67            \\
HE1N      & 2 & 1   & 32       & 0.22              & 32.99              & 22.94               & 88.38              & 8.53           & 3.35            \\
HE1N      & 2 & 1   & 64       & 0.39              & 52.63              & 26.24               & 168.54             & 12.39          & 3.56            \\
HE1N      & 2 & 1   & 128      & 1.2               & 89.01              & 26.18               & 226.2              & 13.16          & 8.1             \\
HE1N      & 2 & 8   & 32       & 0.6               & 57.88              & 25.9                & 177.36             & 11.17          & 7.18            \\
HE1N      & 2 & 8   & 64       & 0.32              & 43.93              & 26.53               & 96.78              & 12.18          & 2.27            \\
HE1N      & 2 & 8   & 128      & 1.13              & 78.11              & 24.42               & 212.75             & 11.07          & 8.4             \\
HE1N      & 2 & 16  & 64       & 0.33              & 53.97              & 27.15               & 168                & 13.67          & 4.47            \\
HE1N      & 2 & 16  & 128      & 0.63              & 68.73              & 25.22               & 194.42             & 11.01          & 7.65            \\
HE1N      & 3 & 1   & 32       & 8.54              & 183.19             & 24.24               & 522.07             & 12.06          & 9.09            \\
HE1N      & 3 & 1   & 64       & 3.67              & 125                & 29.49               & 467.36             & 18.22          & 11.43           \\
HE1N      & 3 & 1   & 128      & 27.84             & 313.76             & 26.94               & 1235.77            & 15.04          & 11.75           \\
HE1N      & 3 & 8   & 32       & 115               & 462.45             & 32.61               & 1556.17            & 21.11          & 19.79           \\
HE1N      & 3 & 8   & 64       & 9.75              & 180.08             & 25.87               & 500.62             & 15.03          & 10.39           \\
HE1N      & 3 & 8   & 128      & 36.05             & 259.15             & 30.1                & 836.27             & 20.68          & 11.45           \\
HE1N      & 3 & 16  & 64       & 30.96             & 378.99             & 28.24               & 1338.33            & 15.51          & 13.3            \\
HE1N      & 3 & 16  & 128      & 8.13              & 226.32             & 27.92               & 621.95             & 18.01          & 10.89           \\
HE2       & 2 & 32  & n/a      & 0.16              & 85.79              & 26.82               & 305.52             & 11.68          & 4.83            \\
HE2       & 2 & 64  & n/a      & 0.17              & 95.92              & 29.71               & 354.79             & 16.9           & 3.26            \\
HE2       & 2 & 128 & n/a      & 0.22              & 132.53             & 32.84               & 540.78             & 22.83          & 4.92            \\
HE2       & 3 & 32  & n/a      & 0.23              & 130.3              & 31.18               & 513.93             & 23.77          & 6.52            \\
HE2       & 3 & 64  & n/a      & 0.29              & 145.62             & 32.84               & 615.9              & 24.61          & 6.3             \\
HE2       & 3 & 128 & n/a      & 0.52              & 249.47             & 29.54               & 1443.82            & 16.56          & 18.34           \\
HE2       & 4 & 32  & n/a      & 0.39              & 175.63             & 29.5                & 733.23             & 20.69          & 6.01            \\
HE2       & 4 & 64  & n/a      & 0.7               & 255.3              & 29.55               & 1578.39            & 18.29          & 16.24           \\
HE2       & 4 & 128 & n/a      & 2.7               & 465.51             & 37.47               & 2943.91            & 22.15          & 15.41           \\
HE2N      & 2 & 1   & 32       & 0.27              & 147.83             & 29.74               & 571.94             & 16.58          & 5.66            \\
HE2N      & 2 & 1   & 64       & 0.43              & 202.74             & 33.36               & 1291.68            & 18.3           & 13.23           \\
HE2N      & 2 & 1   & 128      & 1.58              & 354.19             & 33.76               & 1977.51            & 17.13          & 12.46           \\
HE2N      & 2 & 8   & 32       & 0.59              & 234.83             & 31.42               & 1413.31            & 15.21          & 14.92           \\
HE2N      & 2 & 8   & 64       & 0.33              & 163.78             & 27.42               & 635.64             & 13.6           & 6.18            \\
HE2N      & 2 & 8   & 128      & 0.9               & 307.68             & 36.32               & 1850.83            & 21.71          & 15.79           \\
HE2N      & 2 & 16  & 64       & 0.42              & 208.1              & 29.96               & 1230.56            & 13.41          & 13.16           \\
HE2N      & 2 & 16  & 128      & 0.73              & 274.48             & 30.82               & 1585.1             & 14.85          & 15.04           \\
HE2N      & 3 & 1   & 32       & 5.72              & 651.1              & 36.49               & 3438.96            & 18.67          & 19.05           \\
HE2N      & 3 & 1   & 64       & 4.45              & 477.52             & 35.33               & 3073.46            & 18.75          & 19.77           \\
HE2N      & 3 & 1   & 128      & 26.83             & 1192.79            & 43.23               & 6416.43            & 22.48          & 25.12           \\
HE2N      & 3 & 8   & 32       & 87.38             & 1658.36            & 49.63               & 8139.19            & 23.71          & 27.24           \\
HE2N      & 3 & 8   & 64       & 5.21              & 607.75             & 36.54               & 3337.1             & 22.28          & 17.39           \\
HE2N      & 3 & 8   & 128      & 17.14             & 945.64             & 40.49               & 4620.69            & 25.91          & 22.41           \\
HE2N      & 3 & 16  & 64       & 39.19             & 1368.18            & 44.88               & 7005.7             & 24.1           & 28.3            \\
HE2N      & 3 & 16  & 128      & 11.39             & 774.07             & 36.05               & 3845.1             & 20.29          & 20.74 \\
\bottomrule
\end{tabular}
\end{adjustbox}
\caption{Timings for each experimental configuration. \emph{Init} is the initialisation time for the encryption algorithm, \emph{Enc} is the mean time to encrypt a single integer, \emph{Exec} is the MR job execution time, \emph{Prod} is the mean time to homomorphically compute the product of two encrypted integers, \emph{Sum} is the mean time to homomorphically compute the sum of two encrypted integers.}
\label{results}
\end{table*}

HE1, HE1N, HE2, and HE2N have been implemented in pure unoptimised Java using the JScience mathematics library \cite{jscience2014}. Secure pseudo-random numbers are generated using the ISAAC algorithm \cite{isaac2008}. The ISAAC cipher is seeded using the Linux {\small\texttt{/dev/random}} source. This prevents the weakness in ISAAC shown by Aumasson~\cite{aumasson2006isaac}.

We devised a simple evaluation experiment to generate a fixed (24,000) number of encrypted inputs and then perform a homomorphic inner product on those inputs using a Hadoop MapReduce (MR) algorithm. On the secure client side, the MR input is generated as pseudo-random $\rho$-bit integers which are encrypted and written to a file with $d$ inputs per line, where $d$ is the degree of the polynomial to be computed. In addition, the unencrypted result of the computation is computed so that it may checked against the decrypted result of the homomorphic computation. On the Hadoop cluster side, each mapper processes a line of input by homomorphically multiplying together each input on a line and outputs this product. A single reducer homomorphically sums these products. The MR algorithm divides the input file so that each mapper receives an equal number of lines of input, thereby ensuring maximum parallelisation. Finally, on the secure client side, the MR output is decrypted.

Our test environment consisted of a single secure client (an Ubuntu Linux VM with 16GB RAM) and a Hadoop 2.7.3 cluster running in a heterogeneous OpenNebula cloud. The Hadoop cluster consisted of 17 Linux VMs, one master and 16 slaves, each allocated 2GB of RAM. Each experimental configuration of algorithm, polynomial degree ($d$), integer size ($\rho$), and effective entropy of inputs after adding ``noise'' ($\rho'$, for the `N' variant algorithms only), was executed 10 times. The mean results are tabulated in Table \ref{results}.

Our results compare extremely favourably with Table 2 of \cite{naehrig2011can}. For encryption, our results are, in the best case, 1000 times faster than those presented there, and, in the worst case, 10 times faster. For decryption, our results are comparable. However, to decrypt our results we take the modulus modulo a large primes rather than 2 as in the case of \cite{naehrig2011can}. This is obviously less efficient. For homomorphic sums and products, our algorithms perform approximately 100 times faster. \cite{naehrig2011can} only provides experimental data for computing degree 2 polynomials. We have provided experimental results for the computation of higher degree polynomials.

Similarly, compared with figure 13 of \cite{popa2011cryptdb}, our encryption times for a 32-bit integer are considerably faster. While a time for computing a homomorphic sum on a column is given in figure 12, it is unclear how many rows exist in their test database. Nevertheless, our results for computing homomorphic sums compare favourably with those given. It should be noted that CryptDB \cite{popa2011cryptdb} only supports homomorphic sums and is incapable of computing an inner product. Therefore, we only compare the homomorphic sum timings.

Table 1 of \cite{stephen2014practical} is unclear on whether the values are aggregate timings or the timing per operation. Even assuming that they are aggregate values, our results are approximately 100 times faster than those presented for homomorphic sum and product operations. We also note that Crypsis \cite{stephen2014practical} uses two different encryption schemes for integers, ElGamal \cite{elgamal1985} and Paillier \cite{paillier1999}, which only support addition or multiplication but not both. No discussion of computation of an inner product is made in \cite{stephen2014practical} but we expect that the timings would be considerably worse as data encrypted using ElGamal to compute the products would have to be shipped back to the secure client to be re-encrypted using Paillier so that the final inner product could be computed.

We note that there are some apparent anomalies in the data. JScience implements arbitrary precision integers as an array of Java \texttt{long} (64-bit) integers to store the bit representation of an integer. It may be the case that this underlying representation is optimal for some of our test configurations and suboptimal for others, causing unexpected results. Another possibility is that the unexpected results may be as a result of JVM heap increases and garbage collection which may have been more prevalent in certain test configurations.

\section{Conclusion}
\label{sec:concfurther}

In this paper we have presented several new homomorphic encryption schemes intended for use in a practical SSCC system. We envisage that the majority of computation on integer big data, outside of scientific computing, will be computing low degree polynomials on integers, or fixed-point decimals which can be converted to integers. Our somewhat homomorphic schemes are perfectly suited to these types of computation.

As they are only somewhat homomorphic, each of these schemes has a concern that the computational result will grow bigger than the secret modulus. In the case of the ``noise'' variants, we also have to consider the noise term growing large. So, as they stand, these schemes can only compute polynomials of suitably bounded degree.

A further concern is that the ciphertext space is much larger than the plaintext space. This is as a result of adding multiples of large primes to the plaintext. However, we have shown that values exist which would make the system practical for computing low degree polynomials. Similar schemes \cite{vandijk2010fully,coron2011fully} produce ciphertexts infeasibly larger than the corresponding plaintext, which is a single bit. For example, it should be noted, that even the practical CryptDB \cite{popa2011cryptdb}, which is not fully homomorphic, enciphers a 32-bit integer as a 2048-bit ciphertext. Our schemes will produce ciphertext of similar size, if high security is required. However, if the security is only intended to prevent casual snooping, rather than a determined cryptographic attack, the ciphertext size can be reduced, and the blow-up may be acceptable. Observe that the parameters of the system will change for each computation, so a sustained attack has constantly to re-learn these parameters. Of course, if the attacker is able to export data for off-line cryptanalysis, only high security suffices.

We have also presented a hierarchy of systems, HE$k$, with increasing levels of security. These seem to be of practical interest for small $k>2$, but seem impractical for large $k$.

Finally, we presented a fully homomorphic scheme based on HE$k$ for large enough $k$, which seems of purely theoretical interest. The scheme is capable of computing an arbitrary depth Boolean circuit without employing the techniques  used in other fully homomorphic systems \cite{gentry2009fully,brakerski2012leveled,brakerski2012fully}.

We have implemented and evaluated the HE1, HE1N, HE2 and HE2N schemes as part of an SSCC system as discussed in section \ref{sec:results}. Our results are extremely favourable when compared with \cite{naehrig2011can,popa2011cryptdb,stephen2014practical}. So much so, that our MapReduce job execution times remain low even when using the largest set of parameters for HE2N. We believe that this demonstrates the suitability of our schemes for the encryption of integers in cloud computations.

\printbibliography

\clearpage
\appendix

\section{Proofs}
\factorise*
\begin{proof}
Suppose that we have an unknown plaintext $m$, encrypted as $c = m + r p \mod{pq}$, where $r\uar [1,q)$.

If we can factor $pq$ in polynomial time, we can determine $p$ and $q$ in polynomial time, since we know $p<q$. Therefore, we can determine~$m=c\bmod p$.

If we can determine $m$ given $c$ for arbitrary $m$, then we can determine $rp=c-m$. We are given $qp$, and we know
$0< r < q$, so $\gcd(rp,qp)$ must be $p$, and we can compute $p$ in polynomial time. Now, given $p$, we can determine $q$ as $qp/p$. Hence, we can factorise $pq$ in polynomial time.
\end{proof}

\collision*
\begin{proof}
$\Pr(m_1=m_2)=\sum_{i=0}^{M-1} \xi_i^2= 2^{-H_2}\leq 2^{-\rho}$, since $H_2\geq H_\infty=\rho$.
\end{proof}

\ind*\vspace{0ex}

\begin{proof}\ \vspace{-2ex}
\begin{align*}
	c = m + s\kappa + r p = m + r p \mod{\kappa},
\end{align*}
where $r\uar [1,q)$. Thus, for $i\in[0,\kappa)$,
\begin{align*}
	\Pr\big( c\bmod\kappa=i) &= \Pr(m + rp =i\!\!\mod{\kappa}\big)\\
& = \Pr\big(r= p^{-1}(i-m)\!\!\mod{\kappa}\big)\\
& \in \big\{\lfloor q/\kappa\rfloor 1/q, \lceil q/\kappa \rceil 1/q\big\}\\
& \in [1/\kappa-1/q,1/\kappa+1/q],
\end{align*}
where the inverse $p^{-1}$ of $p$ mod $\kappa$ exists since $p$ is a prime.
Hence the total variation distance from the uniform distribution is
\[ \tfrac12\sum_{i=0}^{\kappa-1} |\Pr\big( c\bmod\kappa=i)-1/\kappa| < \kappa/q .\]
This is exponentially small in the security parameter $\lambda$ of the system, so the distribution of $c\bmod\kappa$ cannot be distinguished in polynomial time from the uniform distribution. Note further that $c_1\bmod\kappa$, $c_2\bmod\kappa$ are independent for any two ciphertexts $c_i=m_i + s_i\kappa + r_i p$ $(i=1,2)$, since  $r_1,r_2$ are independent.

To show IND-CPA, suppose now that known plaintexts $\mu_1,\ldots,\mu_n$ are encrypted by an oracle for HE1N, giving ciphertexts $c_1,\ldots,c_n$. Then, for $r_i\uar[0,q)$, $s_i\uar[0,\kappa)$, we have an SPACDP with ciphertexts $c_i = m_i + s_i\kappa + r_i p$, and the approximate divisor $p$ cannot be determined in polynomial time in the worst case. However, the offsets in this SPACDP are all of the form $\mu_i + s_i\kappa$, for known $m_i$, and we must make sure this does not provide information about $p$. To show this, we rewrite the SPACDP as
\begin{equation}\label{eq:indcpa}
	c_i = \mu_i + s_i\kappa + r_i p = \mu'_i + s'_i\kappa,\quad (i=1,2,\ldots,n),
\end{equation}
where $s'_i=s_i+\lfloor (m_i+r_ip)/\kappa\rfloor$, and $\mu'_i=\mu_i + r_ip\pmod\kappa$.  Now we may view \eqref{eq:indcpa} as an ACDP, with ``encryptions'' $\mu'_i$ of the $\mu_i$, and approximate divisor $\kappa$. Since ACDP is at least as hard as SPACDP, and the offsets $\mu'_i$ are polynomial time indistinguishable from uniform $[0,\kappa)$, from above, we will not be able to determine $\kappa$ in polynomial time.  Now, the offsets $m'_1,m'_2$ of any two plaintexts $m_1,m_2$ are polynomial time indistinguishable from $m'_2,m'_1$, since they are indistinguishable from two independent samples from uniform $[0,\kappa)$.  Therefore, in polynomial time, we will not be able to distinguish between the encryption $c_1$ of $m_1$ and the encryption $c_2$ of $m_2$.
\end{proof}

\hetworandom*
\begin{proof}
	Consider a ciphertext vector which encrypts the plaintext, $m$, and the expression $m+rp+sa\mod{pq}$ which represents one of its elements. Then $r \uar [0,q)$, $s \uar[0,pq)$.

	Consider first $m+sa$. We know that $a^{-1} \mod{pq}$ exists because $a \neq 0$ ($\bmod\ p$ and $\bmod\ q$). Thus, conditional on $r$,
	\begin{align*}
		\Pr[m+rp+sa=i\bmod\ pq] &=\\ \Pr[s = a^{-1}(i-m-rp)&\bmod\ pq]\,=\,\frac{1}{pq}.
	\end{align*}
	Since this holds for any $i\in[0,pq)$, $m+ra+sp \mod{pq}$ is a uniformly random integer from $[0,pq)$.
\end{proof}

\hardness*
\begin{proof}
Suppose we have a system of $n$ approximate prime multiples, $m_i+r_ip$ ($i=1,2,\ldots,n$).  Then
we generate values $a,s_1,s_2,\ldots,s_n\uar[0,pq)$, and we have an oracle set up the cryptosystem with $a_1=a$, $a_2=a+\delta$. The oracle has access to $p$ and provides us with $R$, but no information about its choice of $\varrho$ and $\sigma$. We then generate the ciphertexts $\mbf{c}_i$ $(i=1,2,\ldots,n)$:
\begin{align}
\begin{bmatrix}
 	c_{i1} \\ c_{i2}\end{bmatrix}=\begin{bmatrix}
 	m_i+r_ip+s_ia \\ m_i+r_ip+s_i(a+\delta )\end{bmatrix}\pmod{pq}.\label{eq:170}
\end{align}
Thus $c_{i1}-s_ia=c_{i2}-s_i(a+\delta )=m_i+r_ip$. Thus
finding the $m_i$ in \eqref{eq:170} in polynomial time solves SPACDP in polynomial time.

Conversely, suppose we have any HE2 system with $a_2=a_1+\delta $. The ciphertext for $m_i$ ($i=1,2,\ldots,n$) is
as in \eqref{eq:170}.
so $s_i=\delta ^{-1}(c_{i2}-c_{i1})$. Since $0<\delta <q<p$, $\delta $ is coprime to both $p$ and $q$, and hence $\delta ^{-1}\mod{pq}$ exists. Thus breaking the system is equivalent to determining the $m_i \mod p$ from $m_i+\delta ^{-1}(c_{i2}-c_{i1})a+r_ip$ ($i=1,2,\ldots,n$). Determining the $m_i+\delta ^{-1}(c_{i2}-c_{i1})a$ from the $m_i+\delta ^{-1}(c_{i2}-c_{i1})a+r_ip$ ($i=1,2,\ldots,n$) can be done using SPACDP. However, we still need to determine $a$ in order to to determine  $m_i$. This can be done by ``deciphering'' $R$ using SPACDP. We have
\[ 2\delta ^2\alpha_1= \sigma a-a^2+\varrho p,\qquad 2\delta ^2\alpha_2= \sigma (a+\delta )-(a+\delta )^2+\varrho p,\]
so $\sigma=2\delta ^2(\alpha_2-\alpha_1)-2ka-\delta ^2$. Now $a$ can be determined by first determining $m_0=a(2\delta ^2(\alpha_2-\alpha_1)-(2\delta +1)a-\delta ^2)$ from $m_0+\varrho p=2\delta ^2\alpha_1$. This can be done using SPACDP. Then $a$ can be determined  by solving the quadratic equation $m_0=a(2\delta ^2(\alpha_2-\alpha_1)-(2\delta +1)a-\delta ^2) \mod p$ for $a$. This can be done probabilistically in polynomial time using, for example, the algorithm of Berlekamp \cite{berlekamp1970}.
So the case $\mbf a=[a\ \,a+\delta ]^T$, with known $\delta $, can be attacked using SPACDP on the system
\begin{align*}
 	m_0&+\varrho p,\ m_1+\delta ^{-1}(c_{11}-c_{12})a+r_1p,\\&\ldots,\ m_n+\delta ^{-1}(c_{n1}-c_{n2})a+r_np.\qed
\end{align*}
\end{proof}

\lemten*
\begin{proof}
The $\mbf{a}$'s are a basis if $A_k^{-1}$ exists, since then $\mbf{v}=A_k\mbf{r}$ when $\mbf{r}=A_k^{-1}\mbf{v}$, for any $\mbf{v}$. Now $A_k^{-1}$ exists $\bmod\ {pq}$ if $(\det A_k)^{-1} \mod{pq}$ exists, by constructing the adjugate of $A_k$. Now $(\det A_k)^{-1} \mod{pq}$ exists if $\det A_k\neq 0 \mod{p}$ and $\det A_k\neq 0\mod{q}$.
Now $\det A_k$ is a polynomial of total degree $(k-1)$ in the $a_{ij}$ ($0<i\leq k,0<j<k$), and is not identically zero, since $\det A_k=1$ if $\mbf{a}_i=\mbf{e}_{i+1}$ ($1<i<k$). Also $a_{ij}\uar [0,pq)$ implies $a_{ij}\bmod p\uar [0,p)$ and $a_{ij}\bmod q\uar [0,q)$. Hence, using the Schwartz-Zippel Lemma (SZL)~\cite{moshkovitz2010szl}, we have
$\Pr(\det A_k = 0\bmod{p})\leq (k-1)/p$ and $\Pr(\det A_k = 0\bmod{q})\leq (k-1)/q$, and it follows that
$\Pr(\nexists\,(\det A_k)^{-1}\bmod{pq})\ \leq\ (k-1)(1/p+1/q)$.\qed
\end{proof}

\lemtwenty*
\begin{proof}
We have $\mbf{v}=A_k\mbf{r}$ for some $\mbf{r}\in\mathbb{Z}^k_{pq}$. Then $A_{\star k}=U_kA_k$ and $\mbf{v}_{\star k}=U_k\mbf{v}$, so $R\mbf{v}_\star=RU_k\mbf{v}=RU_kA_k\mbf{r}=RA_{\star k}\mbf{r}=A_{k}\mbf{r}=\mbf{v}$.\qed
\end{proof}

\thmten*
\begin{proof}
We use the same approach as in Lemma~\ref{lem:10}. Thus $A^{\circ2}_{\star k}$ is invertible provided $\det A^{\circ2}_{\star k}\neq 0\mod{p}$ and $\det A^{\circ2}_{\star k}\neq 0\mod{q}$. Let $\mbf{A}$ denote the vector of $a_{ij}$'s, $(a_{ij}:1\leq i\leq k, 1\leq j<k)$. The elements of $A^{\circ2}_{\star k}$ are quadratic polynomials over $\mbf{A}$, except for the first column, which has all 1's, and columns $2,3,\ldots, k$ which are linear polynomials. So $\det A^{\circ2}_{\star k}$ is a polynomial over $\mbf{A}$ of total degree $2\binom{k}{2}+k-1=k^2-1$. Thus, unless $\det A^{\circ2}_{\star k}$ is identically zero as a polynomial over $\mbf{A}$, the SZL~\cite{moshkovitz2010szl} implies
$\Pr(\nexists\,(\det A^{\circ2}_{\star k})^{-1}\bmod{p})\leq(k^2-1)/p$ and $\Pr(\nexists\,(\det A^{\circ2}_{\star k})^{-1}\bmod{q})\leq(k^2-1)/q$.
Therefore we have $\Pr(\nexists\,(\det A^{\circ2}_{\star k})^{-1}\bmod{pq})\leq(k^2-1)(1/p+1/q)$.

It remains to prove that $\det A^{\circ2}_{\star k}$ is not identically zero as a polynomial over $\mbf{A}$ in either $\mathbb{Z}_p$ or $\mathbb{Z}_q$. We prove this by induction on $k$.  Consider $\mathbb{Z}_p$, the argument for $\mathbb{Z}_q$ being identical. Since $\mathbb{Z}_p$ is a field, $\det A^{\circ2}_{\star k}$ is identically zero if and only if it has rank less than $\binom{k+1}{2}$ for all $\mbf{A}$. That is, there exist $\lambda_{ij}(\mbf{A})\in\mathbb{Z}_p$ ($0\leq i\leq j<k$), not all zero, so that
\begin{align*}
\mcl{L}(\mbf{A})\,&=\,\sum_{0\leq i\leq j}^{k-1} \lambda_{ij}\mbf{a}_{\star i}\circ\mbf{a}_{\star j}\\
&\,=\mbf{\alpha}+ \mbf{a}_{\star ,k-1}\circ\mbf{\beta}+ \lambda_{k-1,k-1}\mbf{a}^{\circ2}_{\star ,k-1}\
\,=\, 0,
\end{align*}
where $\mbf{\alpha}=\sum_{0\leq i\leq j}^{k-2} \lambda_{ij}\mbf{a}_{\star i}\circ\mbf{a}_{\star j}$ and
$\mbf{\beta}=\sum_{i=0}^{k-2}\lambda_{i,k-1}\mbf{a}_{\star i}$ are independent of $\mbf{a}_{\star ,k-1}$.

Clearly $\lambda_{k-1,k-1}=0$. Otherwise, whatever $\mbf{\alpha},\mbf{\beta}$, we can choose values for $\mbf{a}_k$ so that  $\mcl{L}\neq0$, a contradiction. Now suppose $\lambda_{i,k-1}\neq 0$ for some $0\leq i<k-1$. The matrix $\hat{A}_\star$ with columns $\mbf{a}_{\star i}$ ($0\leq i<k-1$) contains $A_{k-1}$ as a submatrix, which has rank $(k-1)$ with high probability by Lemma~\ref{lem:10}. Thus $\beta\neq\mbf{0}$ and, whatever $\mbf{\alpha}$, we can choose values for $\mbf{a}_k$ so that  $\mcl{L}\neq0$. Thus $\lambda_{i,k-1}= 0$ for all $0\leq i<k$. Thus $\lambda_{ij}\neq 0$ for some $0\leq i\leq j<k-1$. Now the matrix $\hat{A}_\star^{\circ 2}$ with $\binom{k}{2}$ columns $\mbf{a}_{\star i}\circ\mbf{a}_{\star j}$ $(0\leq i\leq j<k-1)$ contains $A^{\circ 2}_{\star ,k-1}$ as a submatrix, and therefore has rank $\binom{k}{2}$ by induction. Hence $\alpha\neq \mbf{0}$, implying $\mcl{L}\neq0$, a contradiction.
\end{proof}

\lemthirty*
\begin{proof}
From~Lemma~\ref{lem:10}, $\det A = 0$ $\mod{p}$ or $\det A= 0$ $\mod{q}$ with probability at most $(k-1)(1/p+1/q)$. So $\det A$ is not zero or a divisor of zero $\bmod\,pq$.
The entries of $W'$ are random $[0,pq)$, and $\det W'$ is a polynomial of total degree $k$ in its entries. It is a nonzero polynomial, since $W'=I_k$ is possible. Hence, using the SZL~\cite{moshkovitz2010szl}, $\Pr(\det W'=0\bmod p)\leq k/p$ and $\Pr(\det W'=0\bmod q)\leq k/q$. So $\det W'$ is zero or a divisor of zero $\bmod\,pq$ with probability at most $k(1/p+1/q)$. So $\det A\det W'=0$ $\bmod\,pq$ with probability at most $(2k-1)(1/p+1/q)$. So $\det C\neq 0$ with high probability.
\end{proof}

\lemthirtyfive*
\begin{proof}
Note that $C_0=C$ if $m_1=m_2=\cdots=m_k=0$. Since Lemma~\ref{lem:30} holds in that case, the result follows.
\end{proof}

\lemoneforty*
\begin{proof}
Since $\mbf{\gamma}^T\mbf{1}=1$ and $\mbf{\gamma}^T\mbf{a}_i=0$, $i\in[1,k-1]$, $\mbf{\gamma}^T(\mbf{a}_i\cdot\mbf{a}_j)=\mbf{\gamma}^T\big(p\varrho_{ij}\thsp\mbf{1}+\sum_{l=1}^{k-1}\sigma_{ijl}\thsp\mbf{a}_l\big)
=p\varrho_{ij}$. Thus
\begin{align*}
\mbf{a}_i\cdot(\mbf{a}_j\cdot\mbf{a}_l)\ &=\  \mbf{a}_i\cdot\big(p\varrho_{jl}\mbf{1}+\ts\sum_{t=1}^{k-1}\sigma_{jlt}\mbf{a}_t\big) \\
&=\  p\varrho_{jl}\mbf{a}_i+\ts\sum_{t=1}^{k-1}\sigma_{jlt}\mbf{a}_i\cdot\mbf{a}_t,
\end{align*}
and hence
$\mbf{\gamma}^T[\mbf{a}_i\cdot(\mbf{a}_j\cdot\mbf{a}_l)] = p\ts\sum_{t=1}^{k-1}\sigma_{jlt}\varrho_{it}$.
Similarly
$\mbf{\gamma}^T[(\mbf{a}_i\cdot\mbf{a}_j)\cdot\mbf{a}_l] = p\ts\sum_{t=1}^{k-1}\sigma_{ijt}\varrho_{lt}$,
and the lemma follows.
\end{proof}

\lemonefifty*
\begin{proof}
We have $\mbf{\gamma}^T(\mbf{a}_j\cdot\mbf{a}_l)=p\varrho_{ij}=p\varrho_j\varrho_l$ for all $j,\ell\in[1,k-1]$. Hence, $\bmod\ q$,
\begin{align*}
\mbf{\gamma}^T(\mbf{a}_i\cdot(\mbf{a}_j\cdot\mbf{a}_l))\ &=\ p\ts\sum_{t=1}^{k-1}\sigma_{jlt}\varrho_{it}\\
\ &=\ \ p\ts\sum_{t=1}^{k-1}\sigma_{jlt}\varrho_i\varrho_t\\
\ &=\ \ p\varrho_i\ts\sum_{t=1}^{k-1}\sigma_{jlt}\varrho_t\\
\ & =\ p\varrho_i\tau\varrho_j\varrho_l\ =\ p\tau\varrho_i\varrho_j\varrho_l.\qed
\end{align*}
\end{proof}

\section{Derivation of bounds}
\label{app:bounds}
To recap, $n$ is the number of inputs, $M$ is an exclusive upper bound on the inputs, $d$ is the degree of the polynomial we wish to calculate. We take $p \approx 2^\lambda$ and then $q \approx 2^\eta$, where $\eta = \lambda^2/\rho -\lambda$, to guard against the attacks of~\cite{cohn2012approx,howgrave2001approx}.

For HE1, we assume $M\approx 2^\rho$, $n \leq \sqrt{M}$. Therefore,
\[p>(n+1)^dM^d \approx (nM)^d\ \textrm{for\ large}\ n.\]
So, we may take
\begin{align*}
p=2^\lambda &> M^{3d/2}\approx2^{3d\rho/2}\\
\textrm{i.e.}\ \lambda &\approx 3d\rho/2\\
\textrm{and}\ \eta &\approx \frac{\lambda^2}{\rho} -\lambda = \frac{3d\lambda}{2} - \lambda = \frac{3d\rho}{2}\left(\frac{3d}{2} - 1\right)
\end{align*}

For HE1N, we assume $M \approx 2^\rho$, and we have $\rho' = \rho + \lg{\kappa}$. Now,
\begin{align*}
\kappa & >(n+1)^dM^d \approx (nM)^d\ \textrm{for\ large}\ n,\\
\textrm{i.e.}\ \lg{\kappa} &\approx d (\lg{n} + \rho)
\end{align*}
Therefore, since $\rho=\rho'-\lg{\kappa}$,
\begin{align*}
\lg{\kappa} &> d \lg{n} + d(\rho'-\lg{\kappa})\\
\textrm{i.e.}\ \lg{\kappa} &\approx \frac{d(\lg{n} + \rho')}{d+1}
\end{align*}
Since $\kappa$ is much larger than $M$, we also have
\begin{align*}
p=2^\lambda  &> (n+1)^d(M+\kappa^2)^d \approx (n\kappa^2)^d \ \textrm{for\ large}\ n\\
\textrm{i.e.}\ \lambda &\approx d (\lg{n} + 2 \lg{\kappa}),\\
\textrm{and}\ \eta &\approx \frac{\lambda^2}{\rho'} -\lambda = \frac{3d\lambda}{2} - \lambda = \frac{3d\rho'}{2}\left(\frac{3d}{2} - 1\right)
\end{align*}
Then we can calculate $\eta$  as for HE1 above. Note that, in both HE1 and HE1N, $\lambda$ scales linearly with $d$,
and $\eta$ scales quadratically.
\end{document}